\documentclass{techrep}

\usepackage{amsfonts}
\usepackage{amsthm}

\usepackage{soul,color}

\newcommand{\comment}[1]{}

\newtheorem{theorem}{Theorem}{\bfseries}{\itshape}

\newtheorem{lemma}{Lemma}{\bfseries}{\itshape}

{\bfseries}{\itshape}

{\bfseries}{\itshape}

{\bfseries}{\itshape}

{\bfseries}{\itshape}

{\bfseries}{\itshape}

\newtheorem{assumption}{Assumption}{\bfseries}{\rm}

\theoremstyle{definition}

\newtheorem{definition}{Definition}

\newtheorem{rem}{Remark}

{\bfseries}{\rm}

\theoremstyle{plain}

\usepackage{verbatim}
\usepackage{algorithm2e}
\usepackage[utf8]{inputenc}
\usepackage{microtype}
\usepackage{amsthm,amsmath,graphicx}
\usepackage{caption}
\usepackage{subcaption}
\usepackage[letterpaper]{hyperref}
\usepackage[table,dvipsnames]{xcolor}
\definecolor{linkblue}{named}{Blue}
\hypersetup{colorlinks=true, linkcolor=linkblue,  anchorcolor=linkblue,
	citecolor=linkblue, filecolor=linkblue, menucolor=linkblue,
	urlcolor=linkblue} 
\usepackage{wasysym}
\usepackage{graphicx}
\usepackage{caption}
\usepackage{enumitem}
\usepackage{thmtools, thm-restate}
\usepackage{wrapfig}
\usepackage{float}
\usepackage{stackrel}
\usepackage{fixfoot}
\usepackage[utf8]{inputenc}
\usepackage[english]{babel}
\usepackage{mathrsfs}  
\usepackage{venturis}
\usepackage{placeins}
\usepackage{cleveref}
\usepackage[mathlines]{lineno}
\allowdisplaybreaks
\newlength\problemsep
\title{On geodesic disks enclosing many points\footnote{This is an extended version of a paper appearing in Proceedings of WADS 2025.}}

\author{%
  Prosenjit Bose,%
  \thanks{\affil{Carleton University}, 
          \email{jit@scs.carleton.ca, TylerTuttle@cmail.carleton.ca}}\,
  Guillermo Esteban,%
  \thanks{\affil{Universidad de Alcalá},
          \email{\{g.esteban,david.orden\}@uah.es}}\,
  David Orden,%
  \footnotemark[3]\,
  Rodrigo I. Silveira,%
  \thanks{\affil{Universitat Politècnica de Catalunya}, 
          \email{rodrigo.silveira@upc.edu}}\,
  and Tyler Tuttle\footnotemark[2]
}

\begin{document}

\date{}

\maketitle
\begin{abstract}
Let $ \Pi(n) $ be the largest number such that for every set $ S $ of $ n $ points in a polygon~$ P $, there always exist two points $ x, y \in S $, where every geodesic disk containing $ x $ and $ y $  contains $ \Pi(n) $ points of~$ S $. We establish upper and lower bounds for $ \Pi(n)$, and show that $ \left\lceil \frac{n}{5}\right\rceil+1 \leq \Pi(n) \leq \left\lceil \frac{n}{4} \right\rceil +1 $. We also show that there always exist two points $x, y\in S$ such that every geodesic disk with $x$ and $y$ on its boundary contains at least $ \frac{n}{7+\sqrt{37}} \approx \left\lceil \frac{n}{13.1} \right\rceil$ points both inside and outside the disk. For the special case where the points of $ S $ are restricted to be the vertices of a geodesically convex polygon we give a tight bound of $\left\lceil \frac{n}{3} \right\rceil + 1$.  We provide the same tight bound when we only consider geodesic disks having $ x $ and $ y $ as diametral endpoints. We give upper and lower bounds of $\left\lceil \frac{n}{5} \right\rceil + 1 $ and $\frac{n}{6+\sqrt{26}} \approx \left\lceil \frac{n}{11.1} \right\rceil$, respectively, for the two-colored version of the problem. Finally, for the two-colored variant we show that there always exist two points $x, y\in S$ where $x$ and $y$ have different colors and every geodesic disk with $x$ and $y$ on its boundary contains at least $\left\lceil \frac{n}{27.1}\right\rceil+1$ points both inside and outside the disk. 
\end{abstract}

\section{Introduction}
\label{sec:introduction}

Given a set $ S $ of $ n $ points in the plane in general position---no three of them are collinear and no four of them are cocircular---there always exist two points $u, v \in S$ such that any disk that contains $u$ and $v$ also contains a constant fraction of $S$. Neumann-Lara and Urrutia~\cite{neumann1988combinatorial} proved that this constant fraction is at least $ \left\lceil \frac{n-2}{60} \right\rceil $. This bound was improved in a series of papers~\cite{hayward1989note,hayward1989some} culminating in the best known lower bound of about $\frac{n}{4.7}$ by Edelsbrunner et al.~\cite{edelsbrunner1989circles}. Hayward et al.~\cite{hayward1989some} gave an upper bound by constructing a set of $ n $ points such that for every pair of points $u, v$, there exists a disk with $u, v$ on its boundary containing less than $ \left\lceil \frac{n}{4} \right\rceil $ points, thereby showing that $ \left\lceil \frac{n}{4} \right\rceil + 1 $ is an upper bound for the problem. In addition, they studied the problem for point sets in convex position, giving a tight bound of $ \left\lceil \frac{n}{3} \right\rceil + 1 $. Another version of the problem is to consider disks in the plane having the points $ u $ and $ v $ as diametral endpoints~\cite{akiyama1996circles}. In this case, a tight bound of $ \left\lceil \frac{n}{3} \right\rceil + 1 $ for both the convex and non-convex cases was shown.

The best known lower bound of $ \left(\frac{1}{2} - \frac{1}{\sqrt{12}}\right)n + O(1) \approx \frac{n}{4.7} $ was obtained in multiple ways, see~\cite{claverol2021circles,edelsbrunner1989circles,ramos2009depth}. Edelsbrunner et al.~\cite{edelsbrunner1989circles} were the first to show this using techniques related to the order-$k$ Voronoi diagram. Ramos and Viaña~\cite{ramos2009depth} used known results about $j$-facets of point sets in~$\mathbb{R}^3 $. In fact, these results allowed Ramos and Viaña to prove that there is always a pair of points such that any disk through them has at least $ \frac{n}{4.7} $ points both inside and outside the disk.

Generalizations of the original problem to higher dimensions have been studied by Bárány et al.~\cite{barany1989combinatorial} who proved that any set $ S \subset \mathbb{R}^d $ of $n$ points in general position contains a subset~$ A $ of size $ m = \left\lfloor \frac{1}{2} (d + 3) \right\rfloor $ such that, any ball with $ A $ on its boundary, contains at least $ \frac{m!(d - m - 1)!}{d!}n $ other points of $ S $, for $ d > m$. Moreover, Bárány and Larman extended the results in~\cite{neumann1988combinatorial} from Euclidean balls to ellipses in $ \mathbb{R}^2 $ and, more generally, to quadrics in $ \mathbb{R}^d $~\cite{barany1990combinatorial}. In particular, they showed that any set $ S \subset \mathbb{R}^d $ of $n$ points contains a subset $ A $, of size $ \left\lfloor \frac{1}{4}d(d + 3) \right\rfloor + 1 $, with the property that any quadric through $ A $, contains at least $ \frac{n}{s2^{s+1}} $ points of $ S $, where $ s = \frac{(d + 1)(d + 2)}{2} $.

A colored version of the problem has also been studied in the literature: Given a set $ S $ of $\frac{n}{2}$ red and $\frac{n}{2}$ blue points in the plane, there always exists a bichromatic pair of points $u, v\in S$ such that any disk containing $u$ and $v$ contains at least a constant fraction of the $n$ points.
Prodromou~\cite{prodromou2007combinatorial} proved\footnote{Using techniques from Edelsbrunner et al.~\cite{edelsbrunner1989circles}, a PhD thesis written in Spanish and not published elsewhere had claimed before a lower bound of $ \frac{23n}{250} \approx \frac{n}{11} $~\cite{urrutiaproblemas}.}  that any set $ S \subset \mathbb{R}^d $ colored with $ k = \left\lfloor \frac{d+3}{2} \right\rfloor $ colors contains a subset $ A \subset S $ of $ k $ points, one of each color, such that any ball through all points in $ A $ contains at least $ \frac{n}{2k3^k} $ points of $ S $, which gives a fraction $\frac{n}{36}$ for points in the plane. This fraction for the two-dimensional case was later improved in~\cite{claverol2021circles} to $ \left(\frac{1}{2}-\frac{1}{\sqrt{8}}\right)n-o(n) \approx \frac{n}{6.8} $.

We note that all the upper bounds in the Euclidean setting translate directly to the geodesic setting by enclosing the constructions in a large enough triangle. Thus, in this paper, we focus mainly on the lower bounds, and address all variants of the problem in the case where~$ S $ is a set of~$n$ points inside a simple polygon~$ P $. In addition, throughout the paper, we consider that $ n \geq 28 $, so that all our results hold.

{\renewcommand{\arraystretch}{1.2}
\begin{table}
    \centering
    \begin{tabular}{c|c|c|c|c|}
    
    \hline
    \multicolumn{1}{|c|}{Variant} & \multicolumn{2}{c|}{Euclidean setting} & \multicolumn{2}{c|}{Geodesic setting} \\
    \cline{2-5}
    \multicolumn{1}{|c|}{(see Def.~\ref{dfn:variants})} & Lower bound & Upper bound  & Lower bound  & Upper bound \\
       \hline
       \multicolumn{1}{|c|}{$\Pi(n)$}  & $ \lfloor \frac{n}{4.7}\rfloor $~\cite{edelsbrunner1989circles} & $ \lceil \frac{n}{4} \rceil + 1 $~\cite{hayward1989some} & $ \lceil\frac{n}{5}\rceil+1 $ (Thm.~\ref{thm:noverfive}) & $ \lceil \frac{n}{4}\rceil + 1 $ \\
       \hline
       \multicolumn{1}{|c|}{$\overline{\Pi}(n) $} & $ \lceil \frac{n}{3} \rceil + 1 $~\cite{hayward1989some} & $ \lceil \frac{n}{3} \rceil + 1 $~\cite{hayward1989some} & $ \lceil \frac{n}{3} \rceil + 1 $ (Thm.~\ref{lemma: plane diameter ball lower}) & $ \lceil \frac{n}{3} \rceil + 1 $\\
       \hline
       \multicolumn{1}{|c|}{$ \Pi^{in\mbox{-}out}(n) $} & $ \lfloor \frac{n}{4.7}\rfloor $~\cite{ramos2009depth} & $ \lceil \frac{n}{4} \rceil + 1 $~\cite{hayward1989some} & $ \lceil\frac{n}{13.1} \rceil+1$ (Thm.~\ref{thm:noverfive2}) & $ \lceil \frac{n}{4} \rceil + 1 $ \\
       \hline
       \multicolumn{1}{|c|}{$ \Pi^{diam}(n) $}  & $ \lceil \frac{n}{3} \rceil + 1 $~\cite{akiyama1996circles} & $ \lceil \frac{n}{3} \rceil + 1 $~\cite{akiyama1996circles} & $ \lceil \frac{n}{3} \rceil + 1 $ (Thm.~\ref{thm:diamnover3})& $ \lceil \frac{n}{3} \rceil + 1 $\\
       \hline
       \multicolumn{1}{|c|}{$\Pi^{bichrom}(n)$} & $ \lfloor \frac{n}{6.8}\rfloor $~\cite{claverol2021circles} & $\lceil \frac{n}{5} \rceil + 1 $ (Thm.~\ref{thm:upper}) & $\lceil\frac{n}{11.1}\rceil+1$ (Thm.~\ref{thm:novereleven}) & $ \lceil \frac{n}{5} \rceil + 1 $ (Thm.~\ref{thm:upper}) \\
       \hline
       \multicolumn{1}{|c|}{$\overline{\Pi}^{bichrom}(n)$} & $ \lfloor \frac{n}{6.8}\rfloor $~\cite{claverol2021circles} & $\lceil \frac{n}{4} \rceil + 1 $~\cite{claverol2021circles} & $\lceil\frac{n}{11.1}\rceil+1$ (Thm.~\ref{thm:novereleven}) & $ \lceil \frac{n}{4} \rceil + 1 $ \\
       \hline
       \multicolumn{1}{|c|}{$\Pi^{bichrom-in-out}(n)$} & $\lceil\frac{n}{27.1}\rceil+1$ (Thm.~\ref{thm:bichrom-inout}) & $\lceil \frac{n}{4} \rceil + 1$ & $\lceil \frac{n}{27.1}\rceil+1$ (Thm.~\ref{thm:bichrom-inout}) & $\lceil \frac{n}{4} \rceil + 1$ \\
       \hline
    \end{tabular}
    \caption{Summary of our results, comparing the Euclidean and geodesic settings. Note that non-integer values in the table are approximations, and refer to the corresponding reference for the exact value.}
    \label{tab:summary}
\end{table}
}

We summarize all our results and compare them with the Euclidean setting in Table~\ref{tab:summary}. In Section~\ref{sec:comparison}, we present an overview of key properties of shortest paths in $P$ (which we refer to as the {\em geodesic setting}), emphasizing both their similarities to and differences from the Euclidean case. It is these differences that pose challenges we address to extend the results to the geodesic setting. Using these properties, in Section~\ref{sec:every}, we extend the lower bounds from the Euclidean case~\cite{hayward1989note,neumann1988combinatorial} to the geodesic setting. While the upper bounds carry over directly, the lower bounds require new arguments due to the differences outlined in Section~\ref{sec:comparison}. For instance, unlike the Euclidean metric, which has bounded doubling dimension, the doubling dimension of the geodesic metric is not necessarily bounded since it depends on the number of reflex vertices of $P$. Additionally, we establish a tight bound for the case where the point set is in geodesically convex position. In Section~\ref{sec:diametral}, we show that the techniques in~\cite{akiyama1996circles}, where only diametral disks are considered, can be generalized to the geodesic setting. Finally, in Section~\ref{sec:bichromatic} we extend the results in~\cite{prodromou2007combinatorial,urrutiaproblemas} for two-colored points. Although some of the bounds are not optimal, we deliberately include them because they illustrate how classical proof techniques from the Euclidean setting can be adapted to the geodesic framework. These intermediate results clarify which parts of the original arguments remain valid, which require modification, and where genuine geometric difficulties arise. This makes the structural differences between the two settings more apparent and may facilitate further extensions of these techniques.

The problem of studying the largest number of points contained in any disk through two points inside a polygon is distinct from the planar case. As illustrated in Figure~\ref{fig:definitions}, a geodesic disk of radius~$ r $ in a polygon~$ P $ can be strictly contained in a Euclidean disk of radius~$ r $. As such, in this setting, it is unclear whether there always exists a pair of points such that any geodesic disk through them can contain the same number of points as in the Euclidean setting.

\section{Preliminaries and notation} \label{sec:preliminaries} 
A sequence of
points in the plane, $p_1, \ldots, p_n$, forms a {\em polygonal chain} where the
edges of the chain are the segments ${p_ip_{i+1}}$, for $i\in\{1, \ldots, n-1\}$.
A polygonal chain is {\em simple} if consecutive edges intersect at their common
point and non-consecutive segments do not intersect. In a graph theoretic sense, this is a path. A simple polygon is a simple closed polygonal chain, i.e., a cycle. Informally, a polygon $P$
is a {\em weakly-simple} polygon provided that a slight perturbation of the vertices of the polygon results in a simple polygon. See Akitaya et al. \cite{Ak17} for a
formal definition of weakly-simple as well as an algorithm to quickly recognize
weakly-simple polygons. Given a pair of points $ u, v $ in a polygon $ P $, the
\emph{geodesic} (or \emph{shortest}) \emph{path} from $ u $ to $ v$ in $P$ is
denoted by $ g(u,v) $. It is well-known that $g(u,v)$ is a polygonal chain. The
length of this path is the sum of the lengths of its edges and is denoted by $
\lvert g(u,v)\rvert $. 

A subset $P'$ of $P$ is called \emph{geodesically convex} if, for all pairs of points $u,v \in P'$, the geodesic path $g(u,v)$ is contained in $P'$. The \emph{geodesic convex hull} of a set $S$ of points in~$P$ is the intersection of all geodesically convex sets that contain $S$. We say that $S$ is in \emph{geodesically convex position} if every point of $S$ is on the boundary of the geodesic convex hull of $S$. A \emph{geodesic triangle} on three points $ a,b,c \in P $, denoted $ \triangle (a,b,c) $, is a weakly-simple polygon whose boundary consists of $ g(a,b) $, $ g(b,c) $, and $ g(c,a) $. The geodesic paths $ g(a,b) $ and $ g(a,c) $ follow a common route from $a$ until they diverge at a point $a'$ (which may be $a$ itself). Similarly, let $b'$ be the point where $g(b,a) $ and $ g(b,c) $ diverge, and $ c' $ be the point where $ g(c,a) $ and $ g(c,b) $ diverge. The geodesic triangle $ \triangle (a',b',c') $ is a simple polygon that has $ a', b' $ and $ c' $ as its convex vertices. The geodesic triangle $ \triangle (a',b',c') $ is referred to as the \emph{geodesic core} of $ \triangle (a,b,c) $ with the vertices $a', b', c'$ referred to as the {\em core} vertices corresponding to the vertices $a, b, c$, respectively. The geodesic core of $\triangle(a,b,c)$ is denoted as $ \triangledown (a,b,c) $~\cite{bose2021piercing, pollack1989computing}. In addition, $ \angle abc $ denotes the convex angle formed by the two edges of the geodesic core adjacent to the core vertex $ b' $ in $ \triangledown (a,b,c) $. See Figure~\ref{fig:geo}.

A \emph{geodesic disk} centered at $ c \in P $ with radius $ r \geq 0 $ is the set $ D(c,r) = \{x \in P \mid \lvert g(c,x)\rvert \leq r \} $.
In the plane, there always exists a disk through three given points, assuming they are not collinear. When we refer to a disk {\em through} a set of points, we mean that the points are on the boundary of the disk. In a simple polygon, it is no longer true that there always exists a geodesic disk through three given points. This essentially means that given three points in a simple polygon, there does not always exists a point in the polygon that is geodesically equidistant to all three. However, if there is a geodesic disk through three points in a polygon, then it is unique~\cite{aronov1993furthest}. Given two points $u$ and $v$, there always exists a geodesic disk through $u$ and $v$ centered at $ c_{uv} $, the midpoint of $g(u,v)$, called the \emph{diametral geodesic disk}.

A set $ S' $ of at least three points in $ P $ is \emph{geodesically collinear} if $ \exists x, y \in S' $ such that $ S' \subset g(x,y) $, see set $ \{x,c,y\} $ in Figure~\ref{fig:definitions}. Similarly, a set $ S'' $ of at least four points is \emph{geodesically cocircular} if there exists a geodesic disk with $S''$ on its boundary, see set $ \{w,x,y,z\} $ in Figure~\ref{fig:definitions}. A set $ S $ of $ n $ points in a polygon $ P $ is said to be in \emph{general position} if no three points in $ S $ are geodesically collinear and no four points among $ S $ and the vertices of $ P $ are geodesically cocircular. In this paper we only consider sets $S$ in general position in polygons~$P$ with the extra conditions that there are no two points $x,y\in S$ and a reflex vertex $p$ on the boundary of~$P$ such that $y\subset g(x,p)$, and a second condition related to the bisector of a geodesic path that will be explained later in the section. The first assumption is made to simplify the definitions and ensure the correctness of certain constructions.

\begin{figure}[tb]
	\centering
	\includegraphics{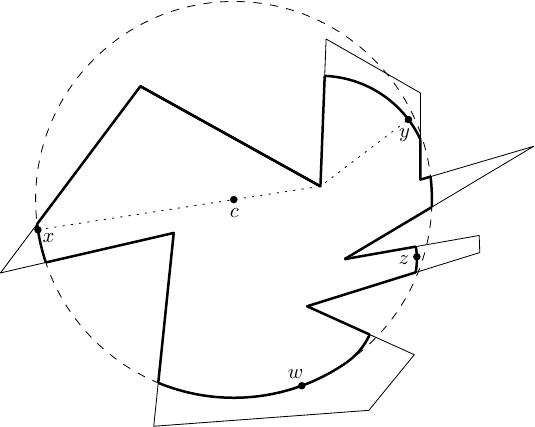}
	\caption{Geodesic disk centered at $ c $, with radius $ \lvert g(c,z)\rvert $, shown solid, and the equivalent Euclidean disk, shown dashed.}
	\label{fig:definitions}
\end{figure}

Given two points $ u $ and $ v $, the \emph{extension path} $ \ell(u,v) $ is defined as follows: extend the first and last segment of $ g(u, v) $ until they intersect with $ P $. Denote these extension points on the boundary as $\hat{u}$ and  $\hat{v}$, respectively. See the red path in Figure~\ref{fig:geo}. Note that if $ u $ (or $ v $) is on the boundary of $ P $ then $\hat{u} = u$ (or $\hat{v} = v$). Also note that, because of the extra condition added to the general position, the extensions cannot touch a reflex vertex $p$ on the boundary of the polygon $P$. The extension path $\ell(u,v)$ naturally partitions the polygon into two weakly-simple polygons. We say that a point $ w $ is to the left (resp., right) of $ g(u,v) $ if $ w $ belongs to the weakly-simple polygon to the left (resp., right) of $ \ell(u,v) $, when considering $ \ell(u,v) $ oriented from $ u $ to $ v $. In addition, we define the left (resp., right) half-polygon of the path $ g(u,v) $ as the set of points that are to the left (resp., right) of its extension path $ \ell(u,v) $ and denote it as $P^+(u,v)$ (resp., $P^-(u,v)).$

Given two points $ u, v \in S $, the bisector $ b(u,v) $ is the set of all points that are equidistant to $ u $ and $v$. In this paper, as in~\cite{aronov1987geodesic}, we assume that no bisector $ b(u,v) $ passes through a vertex of~$ P $. Hence, $ b(u,v) $ is a continuous curve connecting two points on the boundary of~$ P $. Moreover, $ b(u,v) $ can be decomposed into $ O(\lvert P\rvert)$ pieces~\cite{aronov1987geodesic}, each of which is a subarc of a hyperbola (that could degenerate to a line segment).

\begin{figure}[tb]
	\centering
	\includegraphics{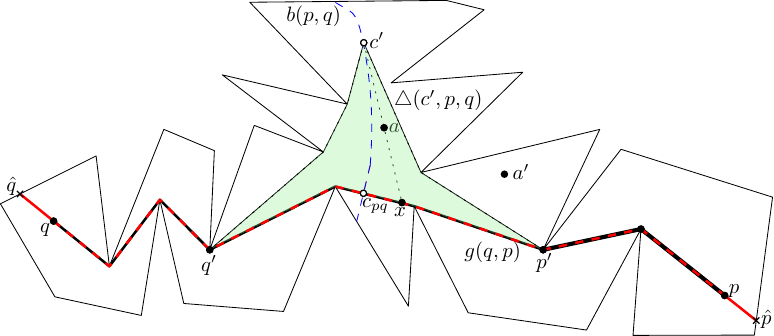}
	\caption{The point $ a $ belongs to the interior of $ \triangle (c',p,q) $.}
	\label{fig:geo}
\end{figure}

\begin{definition} \label{dfn:variants}
    \begin{itemize}
        \item $ \Pi(n) $ (resp., $ \overline{\Pi}(n) $) is the largest integer such that for every polygon $P$ and every set (resp., geodesically convex set) $ S $ of $n$ points in $P$, there exist points $ u, v \in S $ such that every geodesic disk through them contains $ \Pi(n) $ (resp., $\overline{\Pi}(n)$) points of~$ S$.
        
        \item $ \Pi^{in\mbox{-}out}(n) $ is the largest integer such that for every polygon $P$ and every set $ S $ of $n$ points in $P$, there exist points $ u, v \in S $ such that every geodesic disk through them has $ \Pi^{in\mbox{-}out}(n) $ points of $ S $, both inside and outside.
        \item $ \Pi^{diam}(n) $ (resp., $ \overline{\Pi}^{diam}(n) $) is the largest integer such that for every polygon $P$ and every set (resp., geodesically convex set) $ S $ of $n$ points in $P$, there exist points $ u, v \in S $ such that $ D\left(c_{uv},\frac{\lvert g(u,v)\rvert}{2}\right) $ contains $ \Pi^{diam}(n) $ (resp., $\overline{\Pi}^{diam}(n)$) points of $ S$.
        \item $ \Pi^{bichrom}(n) $ (resp., $ \overline{\Pi}^{bichrom}(n) $) is the largest integer such that for every polygon $P$ and every set (resp., geodesically convex set) $S$ of $n/2$ red points and $n/2$ blue points, there exists a red point $r$ and a blue point $b$ such that every geodesic disk through them contains $ \Pi^{bichrom}(n) $ (resp., $ \overline{\Pi}^{bichrom}(n) $) points of $ S $.
        \item $\Pi^{bichrom\mbox{-}in\mbox{-}out}(n)$ is the largest integer such that for every polygon $P$ and every set (resp., geodesically convex set) $S$ of $n/2$ red points and $n/2$ blue points, there exists a red point $r$ and a blue point $b$ such that every geodesic disk through them has $\Pi^{bichrom\mbox{-}in\mbox{-}out}(n)$ points of $ S $, both inside and outside.
    \end{itemize}
\end{definition}

\section{Euclidean and geodesic disks: some similarities and differences}
\label{sec:comparison}

In this section we describe some of the similarities and differences concerning Euclidean and geodesic disks that arise when proving the main results in this paper.

\subsection{Some similarities}

Two bounded curves $ g_1 $ and $ g_2 $ \emph{intersect} if by traversing $ g_1 $ from one of its endpoints to the other endpoint, it crosses $ g_2 $ and switches from one side of $ g_2 $ to the other side~\cite{biniaz2016plane,toussaint1989computing}. We say that $ g_1 $ and $ g_2 $ are \emph{non-intersecting} if they do not intersect. Two non-intersecting curves can share an endpoint or can ``touch'' each other. Note that if $ g_1 $ and $ g_2 $ are geodesics, they can intersect at most once~\cite{biniaz2016plane}, i.e., their intersection consists of at most one subpath of both geodesics.

It is well known that given two segments $ uv $ and $ pq $ that intersect transversely in the plane, any disk through $ u $ and $ v $ contains at least one endpoint of $ pq $, or any disk through $p$ and $q$ contains at least one endpoint of $ uv $, see, e.g., \cite{neumann1988combinatorial}. This result translates to the geodesic setting as shown below. We begin with a property of geodesic disks, refer to Figure~\ref{fig:left-poly}.

\begin{figure}[tb]
	\centering
	\includegraphics{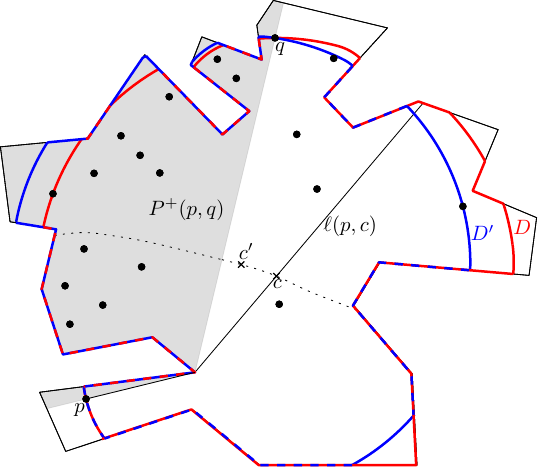}
	\caption{The intersection of the disk through $ p $ and $ q $ with center at $ c $ (in red) and $ P^+(p,q) $ (in gray) is contained in the disk through $ p $ and $ q $ with center at $ c' $ (in blue).}
	\label{fig:left-poly}
\end{figure}

\begin{lemma} \label{lem:diskcontain}
Let $ p, q$ be two points in $ P $ and $ D, D' $ be two geodesic disks through $ p $ and $ q $, respectively with centers $c$ and $c'$. If $c'$ is left of $ \ell(p,c) $ then $D\cap P^+(p,q)\subset D'$.    
\end{lemma}

\begin{proof}
    Let $a$ be an arbitrary point in $D\cap P^+(p,q)$. By definition of $P^+(p,q)$, we have that $a$ is to the left of $\ell(p,q)$. We consider two cases: $a\in \triangle(c',p,q)$ and $a \not \in \triangle(c',p,q)$. 
    
    If $ a $ is in $ \triangle (c',p,q) $, see Figure~\ref{fig:geo}, then we use the fact that, for every point $ x \in g(p,q) $, the length of $ g(c', x) $ as $x$ moves from $p$ to $q$ is a convex function with the maximum occurring at $p$ or $q$~\cite{pollack1989computing}. In particular, we have that $ \lvert g(c',x) \rvert \leq \lvert g(c',q) \rvert =  \lvert g(c',p) \rvert $, for $ x \in g(p,q) $. If we move $ x $ from $ p $ to $ q $ along $ g(p,q) $, at some point $ g(c',x) $ will go through the point $a$, so $ \lvert g(c',a) \rvert \leq \lvert g(c',x) \rvert $. Hence, $ \lvert g(c',a) \rvert \leq \lvert g(c',q) \rvert = \lvert g(c',p) \rvert $, which means that $a \in D'$.
    
    Now consider the case where the point $ a $ is outside $ \triangle (c',p,q) $, like $ a' $ in Figure~\ref{fig:geo}. Without loss of generality, assume $ a $ is to the right of $ g(p,c')$. (Otherwise, we can consider $ a $ and $ q $, and the proof is analogous). In this case, $g(c',p)$ must intersect $g(c,a)$. Let $x$ be a point in this intersection. By definition, we have $ \lvert g(c',x)\rvert + \lvert g(x,p)\rvert = \lvert g(c',p)\rvert$. Since $a \in D$, we have that $\lvert g(c,a)\rvert = \lvert g(c,x) \rvert + \lvert g(x,a) \rvert \le \lvert g(c,p) \rvert $ and by the triangle inequality, we have $\lvert g(c,p) \rvert \le \lvert g(c,x) \rvert + \lvert g(x,p) \rvert $. This means that $\lvert g(x,a) \rvert \le \lvert g(x,p) \rvert $. Therefore, $ \lvert g(c',a) \rvert \leq \lvert g(c',x) \rvert + \lvert g(x,a) \rvert \le \lvert g(c',x) \rvert + \lvert g(x,p) \rvert = \lvert g(c',p) \rvert $, implying $a \in D'$ in this case.
\end{proof}

With this property of geodesic disks in hand, we are now ready to prove the corresponding property of intersecting segments in the geodesic setting.

\begin{lemma}
    \label{lem:quadrilateral}
    Let $u, v, p, q $ be four distinct points in general position in $P$ such that $ g(p,q) $ and $ g(u,v) $ intersect. Then every geodesic disk with $ p $ and $ q $ on its boundary contains at least one endpoint of $ g(u,v) $, or every geodesic disk with $ u, v $ on its boundary contains at least one endpoint of $ g(p,q) $.
\end{lemma}
\begin{proof}
    First, suppose that the bisectors $ b(u,v) $ and $ b(p,q) $ intersect, see Figure~\ref{fig:intersection}, and let $ x $ be a point in their intersection. Then we have that $ \lvert g(x,u)\rvert = \lvert g(x,v)\rvert $ and $ \lvert g(x,p)\rvert = \lvert g(x,q)\rvert $. We can assume, without loss of generality, that $ \lvert g(x,p)\rvert < \lvert g(x,u)\rvert $ since the points in $ S $ are in general position, and $ u,v,p,q $ are not co-circular. The case where $ \lvert g(x,p)\rvert > \lvert g(x,u)\rvert $ is symmetric. 
    
    When $ \lvert g(x,p)\rvert < \lvert g(x,u)\rvert $, the geodesic disk through $ u $ and $ v $ with center $x$ contains $p$ and~$q$. Since $ g(u,v) $ and $ g(p,q) $ intersect, $ p $ and $ q $ are on different sides of the extension path $ \ell(u,v) $. Hence, by Lemma~\ref{lem:diskcontain}, if we move the center of this disk along the bisector $ b(u,v) $ in one direction while keeping $u,v$ on its boundary, the geodesic disk will always contain $ p $ in its interior, and if we move the center in the other direction, the geodesic disk will always contain $ q $.

    Now, suppose that the bisectors $ b(u,v) $ and $ b(p,q) $ do not intersect, see Figure~\ref{fig:non-intersection}. Hence, all points of the bisector $ b(u,v) $ are closer to one endpoint of $ g(p,q) $ than to the other one. Assume, without loss of generality, that this point is $ p $, i.e., $ \lvert g(x',p)\rvert \leq \lvert g(x',q)\rvert, \ \forall x'\in b(u,v) $. Similarly, all points of the bisector $ b(p,q) $ are closer to one endpoint of $ g(u,v) $ than to the other one. Assume, without loss of generality, that this point is $ v $.

    \begin{figure}[tb]
		\captionsetup[sub]{justification=centering}
		\centering
		    \begin{subfigure}[b]{0.47\textwidth}
				\centering
	    	    \includegraphics{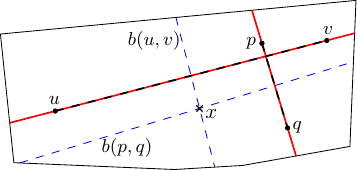}
	    	    \caption{Bisectors $ b(u,v) $ and $ b(p,q) $ intersect.}
	    	    \label{fig:intersection}
	    	\end{subfigure}
                \qquad
		    \begin{subfigure}[b]{0.47\textwidth}
   			    \centering
	    	    \includegraphics{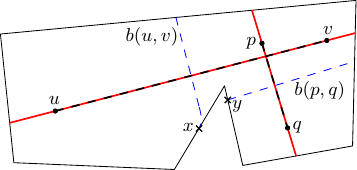}
	    	    \caption{Bisectors $ b(u,v) $ and $ b(p,q) $ do not intersect.}
	    	    \label{fig:non-intersection}
    	    \end{subfigure}
    	    \caption{Any geodesic disk through $ u $ and $ v $ contains at least one endpoint of $ g(p,q)$.}
        \end{figure}

    The bisector $ b(u,v) $ intersects the boundary of the polygon. Also, the extension path $ \ell(u,v) $ divides the polygon $P$ into two half-polygons, one containing $p$ and the other containing $q$. Let $ x $ be the intersection of $ b(u,v) $ with the polygon boundary that is contained in the same half-polygon as $ q $.  Similarly,  the extension path $ \ell(p,q) $ divides the polygon into two half-polygons, and let $ y $ be the intersection of $ b(p,q) $ with the polygon contained in the same half-polygon as $ u $.

    By definition, we have that $ \lvert g(x,u)\rvert = \lvert g(x,v)\rvert $, and $ \lvert g(x,p)\rvert < \lvert g(x,q)\rvert $. If $ \lvert g(x,p)\rvert \leq \lvert g(x,v)\rvert $, the geodesic disk with center at $ x $ and radius $ \lvert g(x,v)\rvert $ contains $ u, v $ and $p$. By Lemma \ref{lem:diskcontain}, this means that every geodesic disk through $u, v$ contains $p$. 

    The only remaining case to consider is when $\lvert g(x,p)\rvert > \lvert g(x,v)\rvert $. Since the bisectors do not intersect, we have that $g(x,v)$ intersects $g(y,p)$. Let $z$ be a point in this intersection. By definition, we have $\lvert g(x,v)\rvert = \lvert g(x,z)\rvert + \lvert g(z,v)\rvert$ and $\lvert g(y,p)\rvert = \lvert g(y,z)\rvert + \lvert g(z,p)\rvert$. By the triangle inequality, we have $\lvert g(x,p)\rvert \leq \lvert g(x,z)\rvert + \lvert g(z,p)\rvert$, which with our assumption that $\lvert g(x,p)\rvert > \lvert g(x,v)\rvert $ implies that $\lvert g(z,v)\rvert < \lvert g(z,p)\rvert$. However, this means that $\lvert g(y,v)\rvert \leq \lvert g(y,z)\rvert + \lvert g(z,v)\rvert < \lvert g(y,z)\rvert + \lvert g(z,p)\rvert = \lvert g(y,p)\rvert$. Therefore, when $\lvert g(x,p)\rvert > \lvert g(x,v)\rvert $, we have $\lvert g(y,v)\rvert < \lvert g(y,p)\rvert$. By Lemma \ref{lem:diskcontain}, this means that every geodesic disk through $p, q$ contains $v$.  
    \end{proof}

We now prove a complementary lemma about points that always remain outside disks through endpoints of geodesic paths that intersect. The proof is similar to that of Lemma~\ref{lem:quadrilateral}, and we include it here for completeness.
\begin{lemma}
    \label{lem:quadrilateral2}
    Let $ u,v,p,q $ be four distinct points in general position in $ P $ such that $ g(p,q) $ and $ g(u,v) $ intersect. Then every geodesic disk with $ p $ and $ q $ on its boundary does {\em not} contain at least one endpoint of $g(u,v)$, or every geodesic disk with $ u, v $ on its boundary does {\em not} contain at least one endpoint of $ g(p,q)$. \end{lemma}

\begin{proof}
Suppose that the bisectors $ b(u,v) $ and $ b(p,q) $ intersect, see
Figure~\ref{fig:intersection}, and let $ x $ be a point in their intersection.
Then we have that $ \lvert g(x,u)\rvert = \lvert g(x,v)\rvert $ and $ \lvert g(x,p)\rvert = \lvert g(x,q)\rvert $. We can
assume, without loss of generality, that $ \lvert g(x,u)\rvert < \lvert g(x,p)\rvert $ since the
points in $ S $ are in general position, and $ u,v,p,q $ are not co-circular.
The case where $ \lvert g(x,u)\rvert > \lvert g(x,p)\rvert $ is symmetric. 
    
When $ \lvert g(x,u)\rvert < \lvert g(x,p)\rvert $, the geodesic disk through $ u $ and $ v $ with
center $x$ contains neither $p$ nor $q$.
Since $ g(u,v) $ and $ g(p,q) $ intersect, $ p $ and $ q $ are on different
sides of the extension path $ \ell(u,v) $. Note that if we exchange the roles of $ p $ by $ q $, and $ D $ by $ D' $ in Lemma~\ref{lem:diskcontain}, we have that $ D' \cap P^-(p,q) \subset D $. Hence, by
Lemma~\ref{lem:diskcontain}, if we move the center of this disk along the
bisector $ b(u,v) $ in one direction while keeping $u,v$ on its boundary, the
geodesic disk will never contain $ p $, and if we move the center in the other
direction, the geodesic disk will never contain $ q $.

Now, suppose that the bisectors $ b(u,v) $ and $ b(p,q) $ do not intersect, see
Figure~\ref{fig:non-intersection}. Hence, all points of the bisector $ b(u,v) $
are closer to one endpoint of $ g(p,q) $ than to the other one. Assume, without
loss of generality, that this point is $ p $, i.e., $ \lvert g(x',p)\rvert \leq \lvert g(x',q)\rvert,
\ \forall x'\in b(u,v) $. Similarly, all points of the bisector $ b(p,q) $ are
closer to one endpoint of $ g(u,v) $ than to the other one. Assume, without loss
of generality, that this point is $ v $.

The bisector $ b(u,v) $ intersects the boundary of the polygon. Also, the
extension path $ \ell(u,v) $ divides the polygon $P$ into two half-polygons, one
containing $p$ and the other containing $q$. Let $ x $ be the intersection of $
b(u,v) $ with the polygon boundary that is contained in the same half-polygon as
$ q $.  Similarly,  the extension path $ \ell(p,q) $ divides the polygon into
two half-polygons, and let $ y $ be the intersection of $ b(p,q) $ with the
polygon contained in the same half-polygon as $ u $.

By definition, we have that $ \lvert g(x,u)\rvert = \lvert g(x,v)\rvert $, and $ \lvert g(x,p)\rvert < \lvert g(x,q)\rvert
$. If $ \lvert g(x,v)\rvert < \lvert g(x,q)\rvert $, the geodesic disk with center at $ x $ through $
u, v $ does not contain $q$. By Lemma \ref{lem:diskcontain}, this means that
every geodesic disk through $u, v$ does not contain $q$. 

The only remaining case to consider is when $\lvert g(x,v)\rvert \geq \lvert g(x,q)\rvert $. By the
fact that $u,v,p,q$ are not co-circular, we have that in fact $ \lvert g(x,v) \rvert > \lvert g(x,q) \rvert $.
Since the bisectors do not intersect, we have that $g(x,q)$ intersects $g(y,v)$.
Let $z$ be a point in this intersection. By definition, we have $\lvert g(x,q)\rvert =
\lvert g(x,z)\rvert + \lvert g(z,q)\rvert$ and $\lvert g(y,v)\rvert = \lvert g(y,z)\rvert + \lvert g(z,v)\rvert$. By the triangle
inequality, we have $\lvert g(x,v)\rvert \leq \lvert g(x,z)\rvert + \lvert g(z,v)\rvert$, which with our
assumption that $\lvert g(x,v)\rvert > \lvert g(x,q)\rvert $ implies that $\lvert g(z,v)\rvert > \lvert g(z,p)\rvert$.
However, this means that $\lvert g(y,q)\rvert < \lvert g(y,z)\rvert + \lvert g(z,q)\rvert \leq \lvert g(y,z)\rvert +
\lvert g(z,v)\rvert = \lvert g(y,v)\rvert$. Therefore, when $\lvert g(x,v)\rvert \ge \lvert g(x,q)\rvert $, we have
$\lvert g(y,p)\rvert < \lvert g(y,v)\rvert$. By Lemma \ref{lem:diskcontain}, this means that every
geodesic disk through $p, q$ does not contain $v$. \end{proof}

We now show two results about the length of the sides of a triangle. These results are trivial in the plane, but in the geodesic setting have not been proven before. In particular, they will be useful in Section~\ref{sec:diametral} to prove the lower bounds on the diametral case.

\begin{lemma}
    \label{prop:geq}
    Let $ \triangle(u,v,w) $ be a geodesic triangle, such that $v$ is a convex vertex of $\triangledown(u,v,w)$ and $ \angle uvw \geq \frac{\pi}{3}$. Then $ \lvert g(u,w)\rvert \geq \min\{\lvert g(u,v)\rvert, \lvert g(v,w)\rvert\}$.
\end{lemma}

\begin{proof}
    Consider the geodesic core $ \triangledown(u,v,w) = \triangle (u',v,w') $ of $ \triangle (u,v,w) $. Recall that $ \angle uvw$ is the angle at vertex $v$ in $ \triangledown(u,v,w) $. Let $ q $ be the intersection of the geodesic $ g(u',w') $ with the line containing the segment of $ g(v,u) $ adjacent to $ v $. Symmetrically, let $ p $ be the intersection of the geodesic $ g(u',w') $ with the line containing the segment of $ g(v,w) $ adjacent to $ v $, see Figure~\ref{fig:proposition2}.

    \begin{figure}[tb]
        \centering  \includegraphics{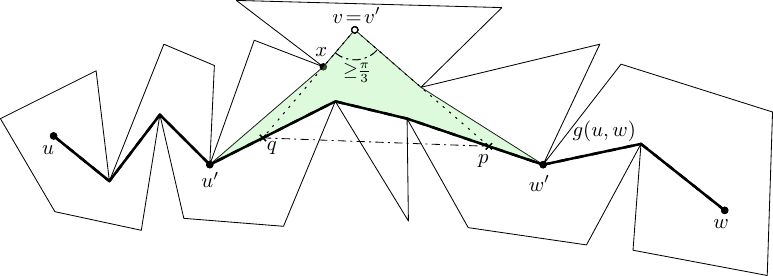}
        \caption{The angle $ \angle uvw $ is greater than $ \frac{\pi}{3} $, and $ \lvert vq\rvert < \lvert vp\rvert $.}
        \label{fig:proposition2}
    \end{figure}

    Since $ \triangle(v,p,q) $ is a Euclidean triangle, we know that $ \lvert pq\rvert \geq
    \min\{\lvert vp\rvert, \lvert vq\rvert\} $. We can assume, without loss of generality, that $ \lvert pq\rvert
    \geq \lvert vq\rvert $. Let $x$ be the vertex that is the endpoint of the first edge of
    $g(v,u)$. By construction $\lvert vq\rvert = \lvert vx\rvert + \lvert xq\rvert$. Note that $\lvert g(p,q)\rvert \ge \lvert pq\rvert
    \ge \lvert vq\rvert = \lvert vx\rvert + \lvert xq\rvert$ (1). By the triangle inequality, $\lvert g(x,u')\rvert \leq \lvert xq\rvert +
    \lvert g(q,u')\rvert$. This implies that $\lvert g(x,u')\rvert - \lvert xq\rvert \leq \lvert g(q,u')\rvert$ (2).
    Putting it all together, we have   
    \begin{align*}
        \lvert g(v,u)\rvert  & =  \lvert g(u,u')\rvert + \lvert g(v,u')\rvert \mbox{ by definition}\\
                  & = \lvert g(u,u')\rvert + \lvert vx\rvert +\lvert g(x,u')\rvert \mbox{ by definition}\\
                  & = \lvert g(u,u')\rvert + \lvert vx\rvert +\lvert g(x,u')\rvert - \lvert xq\rvert + \lvert xq\rvert\\
                  & \le \lvert g(u,u')\rvert + \lvert vx\rvert +\lvert g(q,u')\rvert + \lvert xq\rvert \mbox{ by (2)}\\
                  & \le \lvert g(q,u)\rvert + \lvert g(p,q)\rvert\mbox{ by (1)}\\
                  & \le \lvert g(p,u)\rvert\\
                  & \le \lvert g(w,u)\rvert
    \end{align*}
    
We conclude that $ \lvert g(u,w)\rvert \geq \min\{\lvert g(u,v)\rvert, \lvert g(v,w)\rvert\} $. \end{proof}

The next lemma presents a  result analogous to Lemma~\ref{prop:geq}, but when the angle $ \angle uvw $ is \emph{at most} $ \frac{\pi}{3}$.

\begin{lemma}
    \label{prop:leq}
    Let $ \triangle(u,v,w) $ be a geodesic triangle, such that $ \angle uvw \leq \frac{\pi}{3}$. Then we have that the edge $ \lvert uw\rvert \leq \max\{\lvert g(u,v)\rvert, \lvert g(v,w)\rvert\}$.
\end{lemma}

\begin{proof}
    Since the three points $u, v, w$ form a Euclidean triangle, we know that $ \lvert uw\lvert \leq \max\{\lvert uv\rvert, \lvert vw\lvert\} $. We can assume, without loss of generality, that $ \lvert uw\lvert \leq \lvert uv\rvert $. Since $ \lvert uv\rvert \leq \lvert g(u,v)\rvert $, we thus have that $ \lvert uw\lvert \leq \lvert g(u,v)\rvert \leq \max\{\lvert g(u,v)\rvert, \lvert g(v,w)\rvert\} $.
\end{proof}

\subsection{Some differences}

In contrast to the similarities in the previous section, there are properties that cannot be generalized from the Euclidean  to the geodesic setting. We highlight a few differences. One difference which we alluded to in Section \ref{sec:introduction} is that the geodesic metric, unlike the Euclidean metric, does not necessarily have bounded doubling dimension. Another difference is that the geodesic bisector of a pair of points inside a simple polygon can intersect a line segment more than once~\cite{BoseDD23}. 

The Voronoi diagrams of order $ k $ were one of the tools used to prove the lower bounds in the Euclidean setting, see, e.g.,~\cite{claverol2021circles,edelsbrunner1989circles,ramos2009depth}. However, some important properties of order-$k$ Voronoi diagrams do not extend to the geodesic setting. For instance, in the Euclidean setting, the bisector of a pair of points is a line, whereas in the geodesic setting, the bisector is a finite curve. Moreover, this means that in the Euclidean case, three points in general position uniquely define a disk with those points on the boundary. On the other hand, given three points inside a polygon $P$, there may not exist a point in $ P $ equidistant to all three. Thus, there are cases where three points do not define a geodesic disk with the three points on the boundary~\cite{aronov1993furthest}. 

Another difference between the two settings is that the well-known lifting transformation that maps a point $ (x, y) \in \mathbb{R}^2 $ to the point $ (x, y, x^2+ y^2) \in \mathbb{R}^3 $ used in~\cite{ramos2009depth} to obtain the lower bound of $ \frac{n}{4.7} $ is no longer applicable in the geodesic setting. The exact number of edges of the order-$k$ Voronoi diagram is known in the Euclidean case~\cite{lee1982k}. The formula for this exact number of edges was fundamental to Edelsbrunner et al.'s lower bound proof~\cite{edelsbrunner1989circles}. However, only upper bounds are known in the geodesic case~\cite{bohler2015complexity}. In fact, this is one of the reasons why the lower bounds we obtain in the geodesic setting are not as strong as in the Euclidean setting.

\section{The general case}
\label{sec:every}

In this section we provide a generalization to the geodesic setting of lower bounds that are known for the original problem in the plane. In particular, we generalize three results from~\cite{hayward1989note,hayward1989some,neumann1988combinatorial}. We note that upper bounds for the Euclidean setting are in general also valid for the geodesic setting, since one can take the constructions provided in~\cite{akiyama1996circles,claverol2021circles,hayward1989note}, and enclose the relevant set of points in a bounding polygon, big enough so that the geodesic disks through every three points coincide with the corresponding Euclidean disks.

\subsection{A Constant Fraction of Points Inside Geodesic Disks}
\label{sec:main}

We begin by adapting some definitions and results from~\cite{neumann1988combinatorial} to our setting. The \emph{intersection number} $ I(S) $ of a set of points $ S $ in a polygon is defined as the number of different pairs of geodesic paths $ g(u,v) $, $ g(x,y) $ for $ u,v,x,y \in S $ such that $ g(u,v) $ intersects $ g(x,y) $. Then, the \emph{intersection graph} $ G(S) = (V(G(S)), E(G(S))) $ is defined with $ V(G(S)) = \{g(u,v) \mid u,v \in S, u \neq v\} $, and two vertices $ g(u,v) $, $ g(x,y) $ being joined by an edge in $ G(S) $ if $ g(u,v) $ intersects $ g(x,y) $ and $u,v,x,y$ are four distinct points. Finally, the oriented graph $ \overrightarrow{G}(S) $
of $ G(S) $ is defined orienting an edge $ \{g(u,v), g(x,y)\} $ as $ g(u,v) \rightarrow g(x,y) $ if any geodesic disk through $ u $ and $ v $ contains~$ x $ or $ y $, or otherwise orienting $ g(x,y) \rightarrow g(u,v) $. This orientation is consistent because of Lemma~\ref{lem:quadrilateral}.

\begin{theorem}
    \label{thm:urrutia}
    For any $ n\geq 5 $, $ \Pi(n) \geq \left \lceil\frac{n-2}{60} \right \rceil \approx 0.0166n $.
\end{theorem}

\begin{proof}
    By Kuratowski's Theorem~\cite{kuratowski1930probleme}, each subset of $ 5 $ points of a set $ S $ contains four points $ w, x, y, z $ such that $ g(w,x) \cap g(y,z) \neq \emptyset $. Moreover, the subset $ \{w,x,y,z\} $ appears in exactly $ n - 4$ subsets of $ S $ with five elements. Thus, $ I(S) \geq \frac{\binom{n}{5}}{n-4} $. By the definition of the intersection graph $ G(S) $, $ \lvert E(G(S))\rvert = I(S) $, thus $ \lvert E(G(S))\rvert \geq \frac{\binom{n}{5}}{n-4} $. 
    
    Let $ d^{+}(g(x,y)) $ be the out-degree of a vertex $ g(x,y) $ in $ \overrightarrow{G}(S) $. Then,
    \begin{equation*}
        \sum_{g(x,y)\in V(\overrightarrow{G}(S))} d^{+}(g(x,y)) = \lvert E(G(S))\rvert \geq \frac{\binom{n}{5}}{n-4}. 
    \end{equation*}
    
    Since we have exactly $ \binom{n}{2} $ vertices in $ \overrightarrow{G}(S) $, there is a vertex $ g(u_0,v_0) $ with
    \begin{equation*}
        d^{+}(g(u_0,v_0)) \geq \frac{\frac{\binom{n}{5}}{n-4}}{\binom{n}{2}} = \frac{(n-2)(n-3)}{60}.
    \end{equation*}
    
    This means that any geodesic disk through $ u_0 $ and $ v_0 $ contains at least one endpoint of $ \frac{(n-2)(n-3)}{60} $ geodesic paths $ g(x,y), \ x,y \in S $. Since each point could appear in at most $ n-3$ of these geodesic paths, any geodesic disk through $ u_0 $ and $ v_0 $ contains at least $ \left \lceil \frac{\frac{(n-2)(n-3)}{60}}{n-3}\right \rceil = \left \lceil \frac{n-2}{60} \right \rceil $ points of~$ S $ different from $ u_0 $ and $ v_0 $.
\end{proof}

We now establish an improved lower bound on the number of points contained in any disk with $ x $ and $ y $ on its boundary, by generalizing the results by Hayward in~\cite{hayward1989note} to the geodesic setting. We begin with a preliminary combinatorial result.

\begin{lemma}
    \label{lem:hayward}
    Among any set $ S' $ of $ k $ points in a polygon $ P $, there are at least $ \binom{k-3}{2} $ pairs of points $ x,y \in S' $ such that every geodesic disk through $ x $ and $ y $ contains at least one of the other $ k-2 $ points of $ S' $.
\end{lemma}

\begin{proof}
    Let $ G = (V,E) $ be the graph where $ V\!=\!S' $, and two vertices are joined by an edge if at least one of the geodesic disks through them does not contain any of the other $ k-2 $ points of $ S' $. The edge between $ u $ and $ v $ is drawn as the geodesic path $ g(u,v) $. We claim that~$ G $ is planar, and the edges can be drawn without crossings inside $P$. Suppose there are two edges $ g(u,v), g(x,y) $ in $ G $, for $ u,v,x,y \in S' $, that intersect. By Lemma~\ref{lem:quadrilateral}, one of the edges has the property that every geodesic disk through both of its endpoints contains at least one of the endpoints of the other edge, which contradicts the definition of the edges of $ G $. Thus, $ G $ is a planar graph with $ k $ vertices.
    By Euler's formula, $ \lvert E\rvert \leq 3k-6 $, so $ S' $ contains at least $ \binom{k}{2} - (3k-6) = \binom{k-3}{2} $ pairs that satisfy the property.
\end{proof}

Now, using Lemma~\ref{lem:hayward}, we obtain a tighter lower bound for the number of points contained in any geodesic disk with two points $ u $ and $ v $ on its boundary.

\begin{theorem}
    \label{thm:hayward}
    For any $ n\geq 8 $, $ \Pi(n) \geq \left \lceil\frac{5}{84}(n-2) \right \rceil = \left \lceil\frac{n-2}{16.8}\right\rceil \approx 0.0595n $.
\end{theorem}

\begin{proof}
    Two points in a subset $ S'\subset S $ of $ k $ points \emph{dominate}~$S'$ if every geodesic disk through them contains at least one of the other $ k-2 $ points of~$ S' $. There are $\binom{n}{k}$ subsets of size $k$. By Lemma~\ref{lem:hayward}, each of these sets has $\binom{k-3}{2}$ dominating pairs. The average number of subsets dominated by a pair is $ \frac{\binom{k-3}{2}\binom{n}{k}}{\binom{n}{2}}$. Thus, there is a pair of points $\{u, v\} \in S' $ that dominates at least $ \frac{\binom{k-3}{2}\binom{n}{k}}{\binom{n}{2}}$ subsets of~$ S $.

    For each of these subsets, every geodesic disk through $u$ and $v$ contains at least one of the other $k-2$ points of $ S' $. Since a point of $S\setminus\{u, v\}$ can be in at most $ \binom{n-3}{k-3} $ subsets that include $u$ and $v$, every geodesic disk through $u$ and $ v $ contains at least
    \[\frac{\binom{k-3}{2}\binom{n}{k}}{\binom{n}{2}\binom{n-3}{k-3}} = (n-2)\frac{(k-3)(k-4)}{k(k-1)(k-2)}\]
    other points of $S$. The largest coefficient of $n-2$ is~$\frac{5}{84}$, obtained for $k=8$ and for $k=9$. 
\end{proof}

We now present our strongest lower bound for $\Pi(n)$ by generalizing the approach by Edelsbrunner et al.~\cite{edelsbrunner1989circles} to the geodesic setting. This generalization is made by using order-$k$ geodesic Voronoi diagrams. An \emph{order-$k$ geodesic Voronoi diagram} of a set $ S $ in a polygon $ P $ is the partitioning of $ P $ into cells associated with all subsets $ K \subseteq S $ of size $ k $, such that each cell $ V(K) $ is defined as $ V(K) = \{x \in P \mid \max_{s\in K} g(x,s) \leq \min_{t \in S\setminus K} g(x,t)\} $, i.e., a partition of a polygon into regions such that all points in a region share the same $k$ geodesically nearest sites.

\begin{theorem}
    \label{thm:noverfive}
    For any $n\geq 6$, $\Pi(n) \ge \left\lceil \frac{n}{5} \right\rceil +1 $.
\end{theorem}

\begin{proof}
    Unless specified otherwise, in this proof, order-$k$ diagrams refer to order-$k$ geodesic Voronoi diagrams. Recall that a geodesic disk through two points containing $k-1$ points in its interior has its center on an edge of the order-$k$ diagram. Consider the set of all disks that pass through two points $p,q \in S$ in the polygon. These disks are all centered on the bisector $b(p,q)$, and there are at most $n - 2$ points on the bisector where a third point of $S$ is on the boundary of the disk. Since there is at most one geodesic disk through any three points, this partitions the bisector into {\em at most} $n - 1$ pieces such that any disk through $p$ and $q$ centered on any point belonging to a fixed piece $s$ contains in its interior the same number of points $\rho(s)$. The pieces $s$ such that $\rho(s) = k - 1$ are precisely the edges of the order-$k$ diagram.
    
    For some fixed bisector $b(p,q)$, let $ b_k(p,q) $ be the union of the edges of the order-$1$ through order-$k$ diagrams on $b(p,q)$. The edges of $ b_k(p,q) $ form a number of connected components, where a \emph{connected component} is a maximal connected subset of a bisector that is the union of the closures of the edges of the order-$1$ through order-$k$ diagrams. Note that the number of connected components is strictly less than the total number of edges of all the diagram. Let $\lambda_k $ be the sum of the number of connected components formed by the edges of $ b_k(p,q) $ over \emph{all} pairs $ p,q\in S $. The goal is to find the largest value of $k$ such that $\lambda_k < \binom{n}{2}$. Every bisector is completely covered by a single connected component when $k = n-2$. When $\lambda_k < \binom{n}{2}$, by the pigeonhole principle there exists a bisector $b(u,v)$ that does not have any edges of the order-$1$ through order-$k$ diagrams. Therefore, every geodesic disk with $ u $ and $ v $ on its boundary must contain in its interior at least $k$ points. 

    We derive a formula for $ \lambda_k$ in terms of the number of edges in the order-$1$ to the order-$k$ diagram by considering how the pieces must ``pair up'', in the sense that a piece $s$ with $\rho(s) = i$ can only be adjacent to pieces $ s' $ of the same bisector with $\rho(s') = i \pm 1$. If a piece $ s $ with $ \rho(s) = k-1 $ is adjacent to two pieces $ s_1 $ and $ s_2 $ with $ \rho(s_1) = \rho(s_2) = k-2 $, the number of connected components is decreased by one. If $ s $ is adjacent to only one piece $ s_1 $ with $ \rho(s_1) = k-2 $, it extends an existing component. Otherwise, if it is not adjacent to any piece $ s_1 $ with $ \rho(s_1) = k-2 $, it creates a new connected component. Let $ \epsilon_a, \epsilon_b $ and $ \epsilon_c $ be, respectively, the number of pieces $ s $ with $ \rho(s) = k-1 $ of the first, second, and third type. Then we have 
    \begin{equation}
        \label{eq:lambk}
        \lambda_k = \lambda_{k-1}+\epsilon_c-\epsilon_a.
    \end{equation}
    
    Let $ e_k $ be the number of edges in an order-$k$ diagram and let $S_k$ be the number of intersections between the order-$ k $ diagram and the boundary of the polygon. The edges intersecting the boundary are analogous to the unbounded edges in the Euclidean setting. {Therefore, in this proof we use the term \textit{endpoint} to denote Voronoi vertices of the order-$k$ diagrams, not on the polygon boundary. Hence,} the total number of endpoints of all order-$ k $ edges is given by $ 2e_k - S_k $ since the edges that intersect the boundary of the polygon only contribute one endpoint. Moreover, we note that $ e_k = \epsilon_a+\epsilon_b+\epsilon_c $.
    
     Let $\phi_i$ be the number of unpaired endpoints in $\lambda_i$. Consider all the unpaired endpoints of the connected components of $\lambda_{k-1}$. These endpoints can only come from edges in the order-${k-1}$ diagram since all edges of the order-$i$ diagrams for $i\in \{1,\ldots, k-2\}$ are all paired. Therefore, these unpaired endpoints can only pair up with edges from the order-$k$ diagram, which implies that $ \phi_{k-1} = 2\epsilon_a+\epsilon_b $. Thus, from Equation (\ref{eq:lambk}) and the fact that $e_k=\epsilon_a+\epsilon_b+\epsilon_c$, we have that $ \lambda_k = \lambda_{k-1} + e_k -\phi_{k-1} $. If we iterate over $ k $, we get that $ \lambda_k = (e_k-\phi_{k-1}) + (e_{k-1}-\phi_{k-2}) + \ldots + (e_2-\phi_{1}) + e_1 = \sum_{i=1}^k e_i - \sum_{i=1}^{k-1} \phi_i $. Note that $\phi_{i-1}$ is the number of unpaired endpoints in $\lambda_{i-1}$, which are the paired endpoints of $\lambda_i$. Therefore, $\phi_i+ \phi_{i-1} = 2e_i-S_i $--the number of endpoints in order-$i$ diagrams. We use this to rewrite the equation for $\lambda_k$ in terms of $e_i$ and $S_i$.
    \begin{align*}
        \lambda_k & = \sum_{i=1}^k e_i - [(\phi_{k-1}+\phi_{k-2})+(\phi_{k-3}+\phi_{k-4})+\ldots] \mbox{ (end of sum depends on parity of $k$)} \\
        & = \sum_{i=1}^k e_i - [(2e_{k-1}-S_{k-1})+(2e_{k-3}-S_{k-3})+\ldots+(2e_{2\mbox{-}or\mbox{-}1}-S_{2\mbox{-}or\mbox{-}1})] \mbox{ (depending on parity of $k$)}\\
        & =\sum_{i=1}^k (-1)^{k-i}e_i + \left(S_{k-1}+S_{k-3}+\ldots+S_{2\mbox{-}or\mbox{-}1}\right) \mbox{ ($S_2$ or $S_1$ depending on parity of $k$).}
    \end{align*}
    
    To complete the bound on $\lambda_k$, we need to bound the size of the order-$j$ diagrams. Exact bounds on the size of order-$j$ Euclidean Voronoi diagrams are known but do not apply since unlike the Euclidean case, the geodesic order-$j$ Voronoi diagram has no unbounded edges. We can, however, establish upper and lower bounds using abstract Voronoi diagrams~\cite{Klein89abstractvoronoi}: Geodesic bisectors fulfill all the axioms for an abstract Voronoi diagram except for the property of being unbounded, which in the geodesic setting corresponds to hitting the boundary of the polygon. We can crop this diagram with a closed Jordan curve that encloses all the intersections between bisectors (as is the case for the polygon and the geodesic bisectors) to have an abstract Voronoi diagram~\cite{Papa23}. It is this cropping that no longer allows us to have an exact formula. Since an order-$j$ geodesic Voronoi diagram is an order-$j $ abstract Voronoi diagram, we have that the number of edges $ e_j = 3(F_j-1) - S_j $, where $ F_j = 2jn-j^2-n+1-\sum_{i=1}^{j-1}S_i $ is the number of faces of an order-$j$ geodesic Voronoi diagram, see~\cite[Lemmas~$ 15 $ and $ 16 $]{bohler2015complexity}. Hence, $ e_j = 6jn - 3j^2 - 3n -3\sum_{i=1}^{j-1}S_i - S_j $. Then we get
    \begin{equation}
        \label{eq:edels1}
        e_j - e_{j-1} = 6n-6j+3-2S_{j-1}-S_j.
    \end{equation}
    
    In the final calculation, as in the proof of the lower bound in~\cite{edelsbrunner1989circles}, we distinguish the case that $ k $ is even from $ k $ odd. However, in both cases we have that
    \begin{equation*}
        \lambda_k = 3kn-\frac{3}{2}k^2- \frac{3}{2}k-\sum_{i=1}^k S_i.
    \end{equation*}
    
    We use the fact that $j(j+1) \le \sum_{i=1}^jS_i \le j(2n-j-1)$, see~\cite[Lemma~$ 11 $]{bohler2015complexity}, which implies that 
    \begin{equation*}
        \lambda_k \leq 3kn-\frac{3}{2}k^2- \frac{3}{2}k-k(k+1) = 3kn-\frac{5}{2}k^2- \frac{5}{2}k,
    \end{equation*}
    
    which is a quadratic polynomial in $k$. We note that $\lambda_k <\binom{n}{2}$ whenever $k < \frac{n}{5} $. Thus if we set $k = \left\lceil \frac{n}{5} \right\rceil -1 $, the pigeonhole principle implies that every geodesic disk with $ {u} $ and~$ {v} $ on its boundary must contain in its interior at least $k$ points. Thus, we get that $\Pi(n) \geq {k+2}= \left\lceil\frac{n}{5}\right\rceil +1 $, with the plus two accounting for $ {u} $ and $ {v} $.
\end{proof}

\begin{rem}
    The number of unbounded edges of all order-$ i $ geodesic Voronoi diagrams for all $ i \in \{1, \ldots, j\} $ is at least $ j(j+1) $~\cite{bohler2015complexity}, while the number of unbounded edges of all Euclidean order-$ i $ Voronoi diagrams for all $ i \in \{1, \ldots, j\} $ is at least $ \frac{3}{2}j(j+1) $~\cite{edelsbrunner1989circles}. Thus, the bound obtained in Theorem~\ref{thm:noverfive} is weaker than the one in~\cite{edelsbrunner1989circles}.
\end{rem}

\subsection{A Constant Fraction of Points Inside and Outside Geodesic Disks}

In this section we show that for a point set $ S $, there is a pair of points $ x, y \in S $ such that any disk through them contains a constant fraction of points of $ S $, both inside and outside the disks. We obtain {the first upper} bound by modifying the previous proof of Theorem~\ref{thm:hayward} by using Lemma~\ref{lem:hayward2}. {The result is analogous to Lemma~\ref{lem:hayward} but,} for completeness, we add here the proof of the bound.

\begin{lemma}
    \label{lem:hayward2}
    Among any set $ S' $ of $ k $ points in a polygon, there are at least $ \binom{k-3}{2} $ pairs of points $ x,y \in S' $ such that every geodesic disk through $ x $ and $ y $ leaves outside at least one of the other $ k-2 $ points of $ S' $.
\end{lemma}

\begin{proof}
    Let $ G = (V,E) $ be the graph where $ V\!=\!S' $, and two vertices are joined by an edge if at least one of the geodesic disks through them contains {all} of the other $ k-2 $ points of $ S' $. The edge between $ u $ and $ v $ is drawn as the geodesic path $ g(u,v) $. We claim that~$ G $ is planar, and the edges can be drawn without crossings inside $P$. Suppose there are two edges $ g(u,v), g(x,y) $ in $ G $, for $ u,v,x,y \in S' $, that intersect. By Lemma~\ref{lem:quadrilateral2}, one of the edges has the property that every geodesic disk through both of its endpoints leaves outside at least one of the endpoints of the other {edge}, which contradicts the definition of the edges of $ G $. Thus, $ G $ is a planar graph with $ k $ vertices.
    By Euler's formula, $ \lvert E\rvert \leq 3k-6 $, so $ S' $ contains at least $ \binom{k}{2} - (3k-6) = \binom{k-3}{2} $ pairs that satisfy the property.
\end{proof}

\begin{theorem}
    \label{thm:hayward2}
    For any $ n \ge 21$, $ \Pi^{in\mbox{-}out}(n) \geq \left \lceil\frac{16}{665}(n-2) \right \rceil \approx \left \lceil\frac{n-2}{41.6} \right \rceil\approx 0.024n $.
\end{theorem}

\begin{proof}
    {Let $S'$ be a subset of $S$ of size $k$. Recall from the proof of Theorem~\ref{thm:hayward} that we say two points of $S'$ \emph{dominate} $S'$ if every geodesic disk through them contains at least one of the other $k-2$ points of $S'$.}
    In the same spirit, we define two points in $S'$ to be \emph{dominated by}~$S'$ if every geodesic disk through them leaves outside at least one of the other $ k-2 $ points of $ S' $.

    {There are $ \binom{n}{k} $ subsets of size $ k $ in $ P $.} Among any of these sets there are at most  $ \binom{k}{2} - \binom{k-3}{2} = 3k-6 $ pairs of points that do not satisfy the property of Lemma~\ref{lem:hayward} and at most $ 3k-6 $ pairs of points (possibly different than the previous pairs) that do not satisfy the property of Lemma~\ref{lem:hayward2}. {Since there are $ \binom{n}{2} $ pairs of points in $ S $ and,} for $ k \geq 11 $, it holds that $ 2(3k-6) < \binom{k}{2} $, there is at least one pair of points $ u, v \in S' $ that satisfies both properties, hence dominates and is dominated by at least $ {\frac{\left[\binom{k}{2}-2(3k-6)\right]\binom{n}{k}}{\binom{n}{2}}} $ subsets of~$ S $. For each of these subsets, every geodesic disk through $u$ and $v$ has at least one of the other $k-2$ points of $ S' $, both inside and outside. Since a point of $S\setminus\{u, v\}$ could be in at most $ \binom{n-3}{k-3} $ subsets containing $u$ and $v$, every geodesic disk through $u$ and $ v $ has at least
    \begin{equation*}
        \frac{\left[\binom{k}{2}-2(3k-6)\right]\binom{n}{k}}{\binom{n}{2}\binom{n-3}{k-3}} = (n-2)\frac{2(k-3)(k-4)-k(k-1)}{k(k-1)(k-2)}
    \end{equation*}
    
    other points of $S$, both inside and outside. The largest coefficient {of $\frac{16}{665}$} for $ n-2 $ is obtained when $k=21$.
\end{proof}

{Theorem~\ref{thm:hayward2} illustrates how the argument of~\cite{hayward1989note} can be adapted to the geodesic setting. We improve the result by adapting the arguments from Theorem~\ref{thm:noverfive}.}

\begin{theorem}
    \label{thm:noverfive2}
    For any $n\geq 14$, $\Pi^{in-out}(n) \ge \frac{7-\sqrt{37}}{12}n+O(1)\approx \frac{n}{13.08} $.
\end{theorem}

\begin{proof}
Now consider some fixed bisector $ b(p,q) ${, and let $ b_k(p,q) $ be the union of the edges of the order-$1$ through order-$k$ diagrams on $b(p,q)$}. As in the proof of Theorem~\ref{thm:noverfive}, the edges of $b_k(p,q)$ form a number of connected components (which is strictly less than the total number of edges of all the diagrams). Let $\lambda_{k} $ be the {sum of the} number of connected components formed by the {edges of $ b_k(p,q) $ over \emph{all} pairs $ p,q\in S $}. Analogously, {let $ \overline{b}_k(p,q) $ be the union of the edges of the order-$(n-k)$ through order-$(n-1)$ diagrams on $b(p,q)$. The edges of $ \overline{b}_k(p,q)$ form a number of connected components}. Let $ \overline{\lambda}_{n-k}$ be the {sum of the} number of these connected components formed by the {edges of $ \overline{b}_k(p,q) $ over \emph{all} pairs $ p,q\in S $}. The goal is to find the largest value of~$k$ such that $\lambda_{k} + \overline{\lambda}_{n-k} < \binom{n}{2}$. When $\lambda_k + \overline{\lambda}_{n-k} < \binom{n}{2}$, we note that there exists a bisector $b(p,q)$ that does not have any edges of the order-$1$ through order-$ k $ diagrams, and of the order-$(n-k)$ through order-$(n-1)$ diagrams. Therefore, every geodesic disk with $p$ and $q$ on its boundary has at least $k$ points, both inside and outside the disks.

We already have an upper bound on $\lambda_k$ from the proof of Theorem~\ref{thm:noverfive}, so we will now provide a similar upper bound on $\overline\lambda_{n-k}$. {As in the proof of Theorem~\ref{thm:noverfive}, we use the term endpoint to denote Voronoi vertices of the order-$k$ diagrams, not on the polygon boundary.} We define $\overline\phi_{n-k}$ to be the number of unpaired endpoints after putting down all the edges of the order-$(n-1)$ through the order-$(n-k)$ geodesic Voronoi diagrams. One can follow similar reasoning to the proof of Theorem~\ref{thm:noverfive} to arrive at the formula
\begin{equation*}
    \overline\lambda_{n-k} = \sum_{j=1}^k e_{n-j} - \sum_{j=1}^{k-1} \overline\phi_{n-j}.
\end{equation*}

We also have $\overline\phi_{n-k}+\overline\phi_{n-k+1}=2e_{n-k} - S_{n-k}$. Note that $\overline\phi_{n-1}=2e_{n-1}-S_{n-1}$ since if we only put down edges of the order-$(n-1)$ diagram, every endpoint will be unpaired, and $2e_{n-1}-S_{n-1}$ is the number of endpoints in the order-$(n-1)$ diagram.
Finally, recall that Bohler et al.~\cite[Lemma 11]{bohler2015complexity} proved $j(j+1)\le\sum_{i=1}^j S_i\le j(2n-j-1)$.

As in the proof of Theorem~\ref{thm:noverfive}, we consider the cases where $k$ is even or odd separately. When $k$ is even, we have
\begin{align*}
\overline\lambda_{n-k}
 &= \sum_{j=1}^k e_{n-j} - \sum_{j=1}^{k-1} \overline\phi_{n-j} \\
 &= \sum_{j=1}^k e_{n-j}
	- [(\overline\phi_{n-k+1} + \overline\phi_{n-k+2})
		+ (\overline\phi_{n-k+3} + \overline\phi_{n-k+4})
		+ \ldots
		+ (\overline\phi_{n-3} + \overline\phi_{n-2}) + \overline\phi_{n-1}] \\
 &= \sum_{j=1}^k e_{n-j}
	- [(2e_{n-k+1}-S_{n-k+1})+(2e_{n-k+3}-S_{n-k+3})+\ldots+(2e_{n-1}-S_{n-1})] \\
 &= \sum_{j=1}^k (-1)^{k-j}e_{n-j}
	+ (S_{n-k+1}+S_{n-k+3}+\ldots+S_{n-1}) \\
 &= -\sum_{j=1}^{k/2}(e_{n-2j+1}-e_{n-2j})
	+ \sum_{j=1}^{k/2}S_{n-2j+1} \\
 &= -\sum_{j=1}^{k/2}(6n-6(n-2j+1)+3-S_{n-2j+1}-2S_{n-2j})
	+ \sum_{j=1}^{k/2}S_{n-2j+1} \\
 &= -\sum_{j=1}^{k/2}(6n-6(n-2j+1)+3)
	+ 2\sum_{j=1}^{k/2}(S_{n-2j+1}+S_{n-2j}) \\
 &= -\sum_{j=1}^{k/2}(12j-3)
	+ 2\sum_{j=1}^{k}S_{n-j} \\
 &= -\frac{3}{2}k(k+1)
	+ 2\bigg(\sum_{j=1}^{n-1}S_j - \sum_{j=1}^{n-k-1}S_j\bigg) \\
 &\le -\frac{3}{2}k(k+1)
	+ 2\big(n(n-1) - (n-k)(n-k-1)\big) \\
 &= 4nk - \frac{7}{2}k^2 - \frac{7}{2}k.
\end{align*}

And when $k$ is odd we have
\begin{align*}
\overline\lambda_{n-k}
 &= \sum_{j=1}^k e_{n-j} - \sum_{j=1}^{k-1} \overline\phi_{n-j} \\
 &= \sum_{j=1}^k e_{n-j}
	- [(\overline\phi_{n-k+1} + \overline\phi_{n-k+2})
		+ (\overline\phi_{n-k+3} + \overline\phi_{n-k+4})
		+ \ldots
		+ (\overline\phi_{n-2} + \overline\phi_{n-1})] \\
 &= \sum_{j=1}^k e_{n-j}
	- [(2e_{n-k+1}-S_{n-k+1})+(2e_{n-k+3}-S_{n-k+3})+\ldots+(2e_{n-2}-S_{n-2})] \\
 &= \sum_{j=1}^k (-1)^{k-j}e_{n-j}
	+ (S_{n-k+1}+S_{n-k+3}+\ldots+S_{n-2}) \\
 &= \sum_{j=1}^{(k-1)/2}(e_{n-2j+1}-e_{n-2j})
	+ e_{n-k}
	+ \sum_{j=1}^{(k-1)/2}S_{n-2j} \\
 &= \sum_{j=1}^{(k-1)/2}(6n-6(n-2j+1)+3-S_{n-2j+1}-2S_{n-2j})
	+ e_{n-k}
	+ \sum_{j=1}^{(k-1)/2}S_{n-2j} \\
 &= \sum_{j=1}^{(k-1)/2}(6n-6(n-2j+1)+3)
	+ e_{n-k}
	- \sum_{j=1}^{(k-1)/2} (S_{n-2j+1}+S_{n-2j}) \\
 &= \sum_{j=1}^{(k-1)/2}(12j-3)
	+ e_{n-k}
	- \sum_{j=1}^{k-1} S_{n-j} \\
 &= \frac{3}{2}k(k-1)
	+ e_{n-k}
	- \sum_{j=1}^{k-1} S_{n-j} \\
 &= \frac{3}{2}k(k-1)
	+ (6(n-k)-3)n-3(n-k)^2-S_{n-k}-3\sum_{j=1}^{n-k-1}S_j
	- \sum_{j=1}^{k-1} S_{n-j} \\
 &= \frac{3}{2}k(k-1)
	+ (6(n-k)-3)n-3(n-k)^2
	- \sum_{j=1}^{n-1} S_j - 2\sum_{j=1}^{n-k-1} S_j \\
 &\le \frac{3}{2}k(k-1)
	+ (6(n-k)-3)n-3(n-k)^2
	- n(n-1) - 2(n-k-1)(n-k) \\
 &= 4nk - \frac{7}{2}k^2 - \frac{7}{2}k.
\end{align*}

In both cases we have an upper bound of $\overline\lambda_{n-k}\le 4nk - \frac{7}{2}k^2 - \frac{7}{2}k$.
We can combine this with our upper bound of $\lambda_k \le 3nk - \frac{5}{2}k^2 - \frac{5}{2}k$ from the proof of Theorem~\ref{thm:noverfive} to get
\begin{equation*}
    \lambda_k + \overline\lambda_{n-k} \le 7nk - 6k^2 - 6k.
\end{equation*}

This expression is quadratic in $k$, so it is straightforward to check that $7nk - 6k^2 - 6k < \binom{n}{2}$ whenever
\begin{equation*}
    k < \frac{7n-6-\sqrt{37n^2 - 72n + 36}}{12} = \bigg(\frac{7-\sqrt{37}}{12}\bigg)n + O(1) \approx \frac{1}{13.08}n.
\end{equation*}

Thus, for any $k$ that is less than approximately $ \frac{n}{13.08}$, there exists a pair of points $\{p,q\}$ such that any geodesic disk with $p$ and $q$ on its boundary will contain at least $k$ points and at most $n-k$ points.
\end{proof}

\subsection{Sets in geodesically convex position}

In the Euclidean setting, Hayward et al.~\cite{hayward1989some} proved that when points are in convex position, for every disk containing two points on {its} boundary, there always exists a pair of points such that any disk containing that pair contains at least  $\left\lceil \frac{n}{3}\right\rceil + 1 $ {points}. We generalize that result to the geodesic setting by showing that $ \overline{\Pi}(n) \geq \left\lceil \frac{n}{3}\right\rceil + 1 $. The geodesic case poses its unique challenges, for example, it is not always possible to find an enclosing geodesic disk {(i.e., containing the set~$S$)} with three points on the boundary. However, when no such enclosing geodesic disk exists, we show that there exists a pair of points such that every geodesic disk through them is an enclosing geodesic disk.

\begin{lemma}
    \label{prop:exist1}
    Let $ S $ be a set of points in a polygon $ P $. If there is no enclosing geodesic disk through three points of $S$, then there exist two points of $S$ such that any geodesic disk through them is an enclosing geodesic disk.
\end{lemma}

\begin{proof}
    Define the notation $ D(c)$ to be the disk $ D(c,\max_{p\in S}\lvert g(c,p)\rvert)$ which is the geodesic disk centered at an arbitrary point $c\in P$, and having as radius the maximum geodesic distance from~$ c $ to any other point $ p \in S $. By construction, this is an enclosing geodesic disk. Consider the disk $ D(x) $ centered on {an arbitrary} point $ x \in S $ with a point $ y\in S $ on the boundary of $ D(x) $. Either $D(x)$ has a second point $z$ on its boundary or only has $y$ on its boundary. If the latter is true, then we show how to obtain an enclosing disk with a second point on its boundary. Consider the disk {$D(c')$ as a point $ c' $} moves along $ g(x,y) $ towards $y$. Since the radius is shrinking, eventually, a second point $ z \in S $ will reach the boundary of {$D(c')$}. Note that {$D(c')$ must enclose $S$, with $y$ and $z$ on its boundary.}
    
    {We now consider the point $ c' $ moving along $ b(y,z) $ towards the boundary of the polygon. By definition,} $D(c')$ has $ y $ and $ z $ on its boundary. Moreover, {$D(c')$ must enclose $S$}. If {$D(c')$} is not enclosing, then {there must be a point $ w $ on $b(y, z)$ between $ b(y,z) \cap g(x,y) $ and $ c' $ such that a third point lies on the boundary of $D(w)$}, so we contradict the {assumption} that there is no enclosing geodesic disk through three points. The lemma follows.
\end{proof}

We now generalize the result from Hayward et al.~\cite{hayward1989some} for point sets in convex position to the geodesic setting. In particular, when $S$ is a set of {$n$} points in $P$ that is in geodesically convex position, then there is a pair of points such that any geodesic disk through them contains at least $ \left\lceil \frac{n}{3}\right\rceil + 1 $ points of $S$. 

\begin{theorem}
\label{lemma: plane diameter ball lower}
    For any $n\geq4$, $ \overline{\Pi}(n) \geq \left\lceil \frac{n}{3}\right\rceil + 1 $.
\end{theorem}

\begin{proof}
    By Lemma~\ref{prop:exist1}, if there is no enclosing geodesic disk through three points, there exist two points of $S$ such that any geodesic disk through them is an enclosing geodesic disk. Then all such disks contain all points of~$S$, and the inequality holds.

    Otherwise, there is an enclosing disk $ D $ that has three points on its boundary, namely $x,y,z$ in clockwise order along the boundary of~$D$. Let $ f(\triangle(x,y,z)) $ be the maximum among the number of points from~$S$ contained in each of the three regions $ D \cap P^+(x,y), D \cap P^+(y,z) $ and $ D \cap P^+(z,x) $. {Note that, since $ S $ is geodesically convex, there is no point inside the geodesic core $ \triangledown (x, y, z)$, and each point (other than $ x, y, z$) lies in exactly one of the regions $ P^{+} $.} Among all enclosing disks with three points on the boundary, consider disk $ D_1 $ with points $ u,v, w \in S $ on its boundary which minimizes $ f(\triangle(u,v,w)) $. There must be at least one region among $ D_1 \cap P^+(u,v), D_1 \cap P^+(v,w) $ and $ D_1 \cap P^+(w,u) $ whose interior has at least $ \left\lceil \frac{n-3}{3}\right\rceil $ points. Assume, without loss of generality, that $ D_1 \cap P^+(u,v) $ does{, see Figure~\ref{fig:iterations} (left)}.

    \begin{figure}
        \centering
        \includegraphics{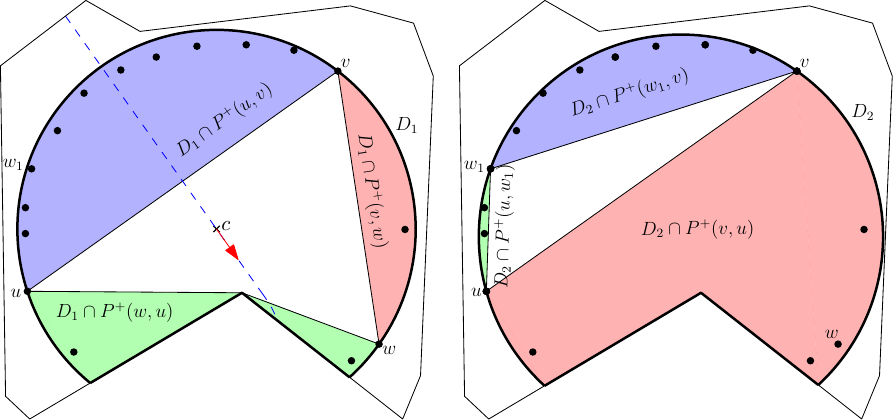}
        \caption{The geodesic disk $ D_2 $ is obtained by moving the center $ c $ of $ D_1 $ along the part of the bisector $ b(u,v)$ (shown dashed) to the right of $ \ell(u,c) $ until hitting point $ w_1 $.}
        \label{fig:iterations}
    \end{figure}

    If the interior of $ D_1 \cap P^-(u,v) $ contains fewer than $ \left\lceil \frac{n-3}{3}\right\rceil $ points, then we continuously change~$D_1$ by the disks through $u,v$ obtained when moving the center~$c$ of $ D_1 $ along the part of the bisector $ b(u,v) $ to the right of $ \ell(u,c) $, as in Figure~\ref{fig:iterations}. 
    As we move $c$, we may arrive at one of two situations: Either the center $c$ hits the boundary of the polygon, or the boundary of the current disk,~$D_2$, hits a point~$w_1\in D_1 \cap P^+(u,v)$. We show, by contradiction, that the latter case cannot occur. 
    If it occurred, the geodesic triangle $ \triangle (u,w_1,v) $ divides $ D_2 $ into three regions $ D_2 \cap P^+(u,w_1)$, $D_2 \cap P^+(w_1,v) $ and $ D_2 \cap P^+(v,u) $. Note that $ D_2 \cap P^+(v,u) = D_2 \cap P^-(u,v)$ contains fewer than $ \left\lceil \frac{n-3}{3}\right\rceil $ points of~$S$, since it is contained in $ D_1 \cap P^-(u,v) $. Also, $ D_2 \cap P^+(u,w_1) $ and $ D_2 \cap P^+(w_1,v) $ both contain fewer points than $ D_1\cap P^+(u,v) $. Thus, $ f(\triangle(u,w_1,v)) < f(\triangle(u,v,w)) $, and contradicting the minimality of $f(\triangle(u,v,w))$.

    Therefore, we will eventually hit the boundary. Then, every geodesic disk through $ u $ and~$ v $ with center in the part of $ b(u,v) $ to the right of $ \ell(u,c) $ contains all the points in $ S $, since the current~$D_1$ is an enclosing disk. Also, by Lemma~\ref{lem:diskcontain}, every geodesic disk through $ u $ and $ v $ with center in the part of $ b(u,v) $ to the left of $ \ell(u,c) $ contains $ D_1 \cap P^+(u,v) $, thus at least $ \left\lceil \frac{n-3}{3}\right\rceil $ points of $S$.  
    
    It follows that there exist two points $u,v$ such that  any geodesic disk through $u$ and $v$ contains at least $ \left\lceil \frac{n-3}{3}\right\rceil $ points of $S$ in its interior. By accounting for $u$ and $v$, on the disk boundary,  we conclude that  $ \overline{\Pi}(n) \geq \left\lceil \frac{n}{3}\right\rceil + 1 $. 
\end{proof}

\section{Diametral geodesic disk}
\label{sec:diametral}
In the Euclidean setting, Akiyama et al.~\cite{akiyama1996circles} showed
that for points in the plane (both in convex position and not) there always exists a diametral disk that contains $\left\lceil \frac{n}{3} \right\rceil + 1$ points of~$S$ and that this is a tight bound.
In this section, we consider the problem of proving lower bounds on the values of $ \Pi^{diam}(n) $ and $ \overline{\Pi}^{diam}(n) $, generalizing the results in~\cite{akiyama1996circles}. We prove that, among a set~$S$ of $n$ points, there is a pair of points such that the diametral geodesic disk through them contains at least $ \left\lceil \frac{n}{3} \right\rceil + 1 $ points of $ S $. 

In order to obtain these lower bounds, we are going to use an enclosing geodesic disk with its center inside the geodesic convex hull of the points on its boundary. 

\begin{lemma}
    \label{prop:exist2}
    Let $ S $ be a set of points in a polygon. There exists an enclosing geodesic disk $ D $ with (i) three points $ x,y,z $ on the boundary, and the center of $ D $ inside $ \triangle(x,y,z)$, or (ii) two points $ x,y $ on the boundary, and the center of $ D $ at the midpoint of $ g(x,y)$.
\end{lemma}

\begin{proof}
    {If there does not exist an enclosing disk with three points on its boundary, then}, by Lemma~\ref{prop:exist1}, there exist two points~$ x, y \in S$ such that every geodesic disk through them is a geodesic enclosing disk. In particular, the disk $ D\left(c_{xy}, \frac{\lvert g(x,y)\rvert}{2}\right) $ is an enclosing diametral disk with $ x $ and $ y $ as its diametral endpoints.
    
    {Suppose therefore that there exists an enclosing disk with three points on its boundary. Let~$ D $ be such a disk, and} let $ x,y,z \in S $ be the three points on the boundary of $ D $, in clockwise order. If the center of the disk is inside the geodesic triangle $ \triangle(x,y,z) $, we are done. Otherwise, the center is inside one of the regions $ D \cap P^+(x,y), D \cap P^+(y,z) $ or $ D \cap P^+(z,x) $. Assume, without loss of generality, it is in $ D \cap P^+(x,y) $. Then we move the center $ c $ of~$ D $ along $ b(x,y) $ towards $ c_{xy} $. 
    If the center $ c $ reaches $ c_{xy} $, then $ D\left(c_{xy}, \frac{\lvert g(x,y)\rvert}{2}\right) $ is a diametral enclosing geodesic disk with $ x $ and $ y $ as its diametral endpoints. 
    If the disk hits a third point $ w\in S $, and the center of the disk is inside $ \triangle(x,y,w) $, we are done. Otherwise, we repeat the process with the points $ x,y, w $ as the new points $ x, y, z $ until the center of the disk through the three points on the boundary of the disk lies inside the geodesic triangle formed by the three points on the boundary. This process ends because the distance {from $c$ to the closest point on the geodesic $g(x,y)$} is decreasing at every iteration.
\end{proof}

\begin{theorem}
\label{thm:diamnover3}
    For any $n\geq3$, $ \Pi^{diam}(n) = \overline{\Pi}^{diam}(n) \geq \left\lceil \frac{n}{3} \right\rceil + 1 $.    
\end{theorem}

\begin{proof}    
     Consider an enclosing disk $D$ as guaranteed by Lemma~\ref{prop:exist2}. Either (i)~$D$ has exactly three points of $S$ on the boundary, and the center of the geodesic disk is inside the geodesic triangle formed by the three points, or (ii)~$D$ has two points of $S$ as diametral endpoints. 
    For case~(ii), $ D $ having two points on the boundary, the claimed result holds, since $D$ contains all the points.

    For case~(i), let $ a, b $ and $ c $ be the three points on the boundary of the geodesic disk. We will prove that any point $ x \in S $ is contained in at least one of the geodesic disks $ D\left(c_{ab},\frac{\lvert g(a,b)\rvert}{2}\right) $, $ D\left(c_{bc},\frac{\lvert g(b,c)\rvert}{2}\right) $ and $ D\left(c_{ca},\frac{\lvert g(c,a)\rvert}{2}\right) $.

    First, consider the case where $ x $ is outside the geodesic triangle $ \triangle (a,b,c) $. Then, $ x $ belongs to one of the regions $ D\cap P^+(a,b), D\cap P^+(b,c) $, and $ D\cap P^+(a,c) $. Without loss of generality, assume $ x \in D \cap P^+(a,b) $. By Lemma~\ref{prop:exist2}, the center of the geodesic disk is inside $ \triangle (a,b,c) $, so it does not belong to $ D \cap P^+(a,b) $. Then, by Lemma~\ref{lem:diskcontain}, $ x $ is contained in the diametral geodesic disk $ D\left(c_{ab},\frac{\lvert g(a,b)\rvert}{2}\right) $.

    Now, we prove the case where $ x \in \triangle (a,b,c) $. Consider the geodesic paths $ g(a,x), g(b,x) $ and $ g(c,x) $. One of the angles $ \angle axb, \angle bxc $ or $ \angle axc $ is greater than or equal to $ \frac{2\pi}{3} $
    . Assume, without loss of generality, that $ \angle bxc \geq \frac{2\pi}{3} $, see blue angle in Figure~\ref{fig:diametrical}. Let $ o $ be the midpoint of the shortest path $ g(b,c) $. We also know that one of the angles $ \angle bxo $ or $ \angle oxc$ is greater than or equal to $ \frac{\pi}{3} $. Assume, without loss of generality, that it is the angle $ \angle oxc $. When $ o $ is not visible from $ x $, the angle $ \angle oxc $ is equal to the angle $ \angle bxc \geq \frac{2\pi}{3} > \frac{\pi}{2} $. Hence, by~\cite[Corollary~$ 2 $]{pollack1989computing}, $ \lvert g(x,o)\rvert < \lvert g(o,c)\rvert = \lvert g(o,b)\rvert $, and $ x $ is contained in the geodesic disk $ D\left(c_{bc},\frac{\lvert g(b,c)\rvert}{2}\right) $. In the case of $ o $ being visible from $ x $, by Lemma~\ref{prop:geq}, $ \lvert g(o,c)\rvert \geq \min\{\lvert g(c,x)\rvert, \lvert g(x,o)\rvert\} $ and, by Lemma~\ref{prop:leq}, $ \lvert g(x,o)\rvert \leq \max\{\lvert g(c,x)\rvert, \lvert g(o,c)\rvert\} $, hence $ \lvert g(x,o)\rvert \leq \lvert g(o,c)\rvert = \lvert g(o,b)\rvert $, and $ x $ is contained in the geodesic disk $ D\left(c_{bc},\frac{\lvert g(b,c)\rvert}{2}\right) $.

    \begin{figure}[tb]
	\centering
	\includegraphics{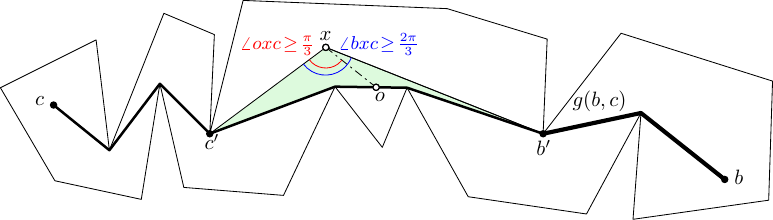}

    \caption{The point $ o $ is visible from $ x $.}
	\label{fig:diametrical}
\end{figure}
    
    Then, the union of the geodesic disks $ D\left(c_{ab},\frac{\lvert g(a,b)\rvert}{2}\right), D\left(c_{bc},\frac{\lvert g(b,c)\rvert}{2}\right) $ and $ D\left(c_{ca},\frac{\lvert g(c,a)\rvert}{2}\right) $ covers all the points in $ P $. Since the points $ a, b $ and $ c $ are counted twice,
    \begin{equation*}
        \Bigg\lvert D\left(c_{ab},\frac{\lvert g(a,b)\rvert}{2}\right) \cap P\Bigg\rvert + \Bigg\lvert D\left(c_{bc},\frac{\lvert g(b,c)\rvert}{2}\right) \cap P\Bigg\rvert + \Bigg\lvert D\left(c_{ca},\frac{\lvert g(c,a)\rvert}{2}\right) \cap P\Bigg\rvert > n + 3.
    \end{equation*}
    
    By the pigeonhole principle, the result follows.
\end{proof}

\section{Bichromatic geodesic disks}
\label{sec:bichromatic}

In this section, we consider the bichromatic version of the problem, which has also received attention in the Euclidean setting. 
Let $B$ be a set of $\frac{n}{2}$ blue points and $R$ be a set of $\frac{n}{2}$ red points inside a polygon $P$, and let $S = B \cup R$.

\subsection{{A Constant Fraction of Bichromatic Points Inside Geodesic Disks}}

Following the procedure in~\cite{urrutiaproblemas}, we obtain a similar result to Theorem~\ref{thm:urrutia}, and we prove that it is possible to find a bichromatic pair such that any geodesic disk containing them contains at least $ \left\lceil \frac{{n}-2}{72}\right\rceil $ points of $S$.

\begin{theorem}
    \label{thm:urrutiabichromatic}
    For any $ n\geq 6 $, $ \Pi^{bichrom}(n) \geq \left \lceil\frac{n-2}{72} \right \rceil \approx 0.0139n$.
\end{theorem}

\begin{proof}
    {Analogously to Section~\ref{sec:main}, we define the intersection graph for the bichromatic setting as follows.} Let $ G(S) = (V(G(S)), E(G(S))) $ be the intersection graph defined as follows: $ V(S) = \{g(u,v) \mid u,v \in S, u \text{ is red, and } v \text{ is blue}\} $, and two vertices $ g(u,v), g(x,y) $ are joined by an edge in $ G(S) $ if {$ g(u,v) $ intersects $ g(x,y) $ and $ u,v,x,y $ are four distinct points}. The oriented graph $ \overrightarrow{G}(S) $ of $ G(S) $ is defined orienting an edge $ {\{g(u, v), g(x, y)\}} $ as $ g(u, v) \rightarrow g(x, y) $ if any geodesic disk through $ u $ and $ v $ contains $ x $ or $ y $, or otherwise orienting $ g(x, y) \rightarrow g(u, v) $. By Lemma~\ref{lem:quadrilateral}, if two of these paths intersect, then one of the two bichromatic pairs fulfills the property that any geodesic disk through it will contain one of the other two endpoints, so this orientation is consistent. 
    
    We now obtain the minimum number of intersections between all pairs of geodesic paths. Let, as in Theorem~\ref{thm:urrutia}, $ I(S) $ be the number of pairs of geodesic paths that intersect. For each $ S' \subset S $ with exactly $ 3 $ blue and $ 3 $ red points, by Kuratowski's Theorem~\cite{kuratowski1930probleme}, there are two blue points $ w, x \in S' $ and two red points $ y, z \in S' $ such that $ g(w,y) \cap g(x,z) \neq \emptyset $. There are $ \dbinom{\frac{n}{2}}{3}^2 $ subsets of this type. Furthermore, every intersection belongs to exactly $ \left(\frac{n}{2}-2\right)^2 $ subsets of this type. Dividing by this number, we remove repetitions, and we obtain that $ I(S) \geq \left(\frac{\binom{\frac{n}{2}}{3}}{\frac{n}{2}-2}\right)^2$.

    By definition of the intersection graph $ G(S) $, $ \lvert E(G(S))\rvert = I(S) $ and $ \lvert V(G(S))\rvert = \left(\frac{n}{2}\right)^2 $. Hence, there exists one vertex $ g(u_0,v_0) \in \overrightarrow{G} $ with out-degree at least $ \frac{I(S)}{\left(\frac{n}{2}\right)^2} $.

    Then, we have that $ \displaystyle
    \frac{I(S)}{\left(\frac{n}{2}\right)^2} \geq \left(\frac{\frac{n}{2}-1}{6}\right)^2 $. This means that there exist two points $ u_0 $ and~$ v_0 $, one of each color, such that the geodesic path between them intersects at least $ \left(\frac{\frac{n}{2}-1}{6}\right)^2 $ other geodesic paths. Since every point is the endpoint of at most $ \frac{n}{2}-1 $ geodesic paths {(with the minus one accounting for $ u_0 $ or $ v_0 $)}, every geodesic disk through this pair of points contains at least $ \left(\frac{\frac{n}{2}-1}{6}\right)^2 \frac{1}{\frac{n}{2}-1} = \frac{\frac{n}{2}-1}{36} = \frac{n-2}{72}$ points of $ S $.
\end{proof}

It was proved in~\cite{prodromou2007combinatorial} that any set of bichromatic points $ S $, contains a pair $ \{p,q\} \subset S $, one point of each color, such that any disk through $ p $ and $ q $ contains a positive fraction of the points of~$ S $, in particular $ \frac{1}{36} $. {Although} the analogous colored result in the geodesic setting {does not give the optimal bound, we present it here to reflect the geodesic translation of a classical technique.}

\begin{theorem}
    \label{thm:prodroumou}
    For any $ n\geq 6 $, $ \Pi^{bichrom}(n) \geq \left \lceil\frac{n-2}{36} \right \rceil +2 \approx 0.0277n$.
\end{theorem}

\begin{proof}
    Let $B$ be the set of blue points, and let $R$ be the set of red points. By Kuratowski's Theorem and Lemma~\ref{lem:quadrilateral}, for every set $ Z \subset S $ of six points, three being blue and three being red, there is a pair $ \{u,v\} $ of bichromatic points such that any geodesic disk with $ u $ and $ v $ on its boundary contains at least one point of $ Z \setminus \{u,v\} $. Hence, there is a family $ \mathcal{Z} $ and a pair $ \{u,v\} \subset S $ of bichromatic points such that $ \{u,v\} $ belongs to every $ Z \in \mathcal{Z} $ and
    \begin{equation}
        \label{eq:lower}
        \lvert \mathcal{Z}\rvert \geq \frac{\dbinom{\frac{n}{2}}{3}^2}{\dbinom{\frac{n}{2}}{1}^2} = \frac{\left(\frac{n}{2}-2\right)^2\left(\frac{n}{2}-1\right)^2}{6^2}.    
    \end{equation}
    
    Now we will upper-bound the value $ \lvert \mathcal{Z}\rvert $. Let $ D $ be a geodesic disk with $ u $ and $ v $ on its boundary, and $ m-2 $ points in its interior. Each $ Z \in \mathcal{Z} $ contains a point of $ D \setminus \{u,v\} $. The set $ \{u,v\} $ can be extended to $ Z $, given that $ D \cap (Z \setminus \{u,v\}) \neq \emptyset $, by (i) choosing a point of~$ Z $ from the remaining $ m-2 $ points in $ D $, (ii) choosing a point of the same color class from the remaining $ \frac{n}{2}-2 $ points in $ S $, and (iii) choosing two points in $ S $ from the remaining color class. Thus, we get that
    \begin{equation}
        \label{eq:upper}
        \lvert \mathcal{Z}\rvert \leq \dbinom{m - 2}{1}\dbinom{\frac{n}{2}-2}{1}\dbinom{\frac{n}{2}-1}{2} = (m-2)\frac{\left(\frac{n}{2}-2\right)^2\left(\frac{n}{2}-1\right)}{2}.
    \end{equation}
    
    Finally, from equations (\ref{eq:lower}) and (\ref{eq:upper}) we get
    $\displaystyle
        m-2 \geq \frac{\frac{n}{2}-1}{18} \Longrightarrow m \geq \frac{\frac{n}{2}+35}{18} = \frac{n+70}{36}.
    $
\end{proof}

{We conclude by adapting the approach of Edelsbrunner et al.~\cite{edelsbrunner1989circles}. Among the classical techniques considered in this section, this method turns out to produce the strongest lower bound in the geodesic setting.}

\begin{theorem}
    \label{thm:novereleven}
    For any $ n\geq12 $, $ \Pi^{bichrom}(n) \geq \left\lceil\frac{n}{6+\sqrt{26}}\right\rceil + 1 \approx \left\lceil\frac{n}{11.1}\right\rceil \approx 0.09n$.
\end{theorem}

\begin{proof}
Similar to the proof of Theorems~\ref{thm:noverfive} and \ref{thm:noverfive2}, we want to find a bisector that does not contain an edge of the order-$1$ through order-$k$ geodesic Voronoi diagrams, however we only consider bisectors $b(p,q)$ where $p$ and $q$ are of different color. {Let $ b_k(p,q) $ be the union of the edges of the order-$ 1 $ through order-$k $ diagrams on $ b(p,q) $.} The edges of $b_k(p,q)$ form a number of connected components (which is strictly less than the total number of edges of all the diagrams). Let $\lambda_{k} $ be the {sum of the} number of connected components {formed by the edges of $ b_k(p,q) $ over \emph{all} $p,q \in S $ of different color}. The goal is to find the largest value of $k$ such that $\lambda_{k} < \left(\frac{n}{2}\right)^2$. Since every bisector must be covered by one connected component, when $\lambda_k < \left(\frac{n}{2}\right)^2 $, we note that there exists a bisector $b(p,q)$ that does not have any edges coming from the order-$1$ through order-$ k $ geodesic Voronoi diagrams. Therefore, every geodesic disk with $p$ and $q$ on its boundary contains at least $k$ points.

By considering how the pieces must ``pair up'', in the sense that a piece $s$ with $\rho(s) = i$ can only be adjacent to pieces $ s' $ with $\rho(s') = i \pm 1$, one can follow the reasoning from Theorem~\ref{thm:noverfive} to show that $ \lambda_k \leq 3kn-\frac{5}{2}k^2- \frac{5}{2}k$, which is a quadratic polynomial in $k$ being strictly less than the number $ \left(\frac{n}{2}\right)^2$ of geodesic bichromatic bisectors whenever $ k < \frac{n}{6+\sqrt{26}} $, therefore the same arguments at the end of the proof of Theorem~\ref{thm:noverfive} lead to $\Pi^{bichrom}(n) \geq \left\lceil\frac{n}{6+\sqrt{26}}\right\rceil+1 $.
\end{proof}

\subsection{{A Constant Fraction of Bichromatic Points Inside and Outside Geodesic Disks}}

In this section, we prove a bound on the number of points both inside and outside the disk in the geodesic setting. This bound has not been shown in the Euclidean setting, however, its proof is similar to the proof of Theorem \ref{thm:noverfive2}.

\begin{theorem}
    \label{thm:bichrom-inout}
    For any $ n\geq 28$, $\Pi^{bichrom\mbox{-}in\mbox{-}out}(n) \ge \left\lceil\frac{n}{14+2\sqrt{43}}\right\rceil+1 \approx \left\lceil\frac{n}{27.1}\right\rceil \approx 0.0369n$.
\end{theorem}

\begin{proof}
    Following the proof of Theorem~\ref{thm:noverfive2}, we consider {the union of the} edges of both the order-$1$ through order-$k$ geodesic Voronoi diagrams, and the {union of the} edges of the order-$(n-k)$ through order-$(n-1)$ geodesic Voronoi diagrams. Recall the upper bound of $\lambda_k+\overline\lambda_{n-k} \le 7kn-6k^2-6k$ on the total number of connected components formed by these edges. Then as in the proof of Theorem~\ref{thm:novereleven}, we consider only bisectors of points $p$ and $q$ of different color. There are $(\frac{n}{2})^2$ such bisectors, and if $\lambda_k+\overline\lambda_{n-k}<(\frac{n}{2})^2$, then there exists at least one bisector which does not contain an edge of the order-$1$ through order-$k$ geodesic Voronoi diagrams, or an edge of the order-$(n-k)$ through order-$(n-1)$ geodesic Voronoi diagrams. This holds as long as $k<\frac{n}{14+2\sqrt{43}}$, and the theorem follows.
\end{proof}

\subsection{Upper bound}

In this section we prove the following upper bound on the number of points contained in any geodesic disk with a pair of bichromatic points on its boundary.

\begin{theorem}
    \label{thm:upper}
    For any $ n\geq6 $, $ \Pi^{bichrom}(n) \leq \left\lceil \frac{n}{5}\right\rceil + 1 $.
\end{theorem}

\begin{proof}
    We show how to construct a configuration of $ n $ bichromatic points $ S $ so that for every pair of bichromatic points $ u, v \in S $, $ \Pi^{bichrom}(n) \leq \left\lceil \frac{n}{5}\right\rceil+1 $. 
    Draw an equilateral triangle with sides of unit length and vertices labeled $ x, y $, and $ z $. Let $ s_1, s_2 $ and $ s_3 $ be, respectively, the midpoints of edges $ {yz}, {xz} $ and $ {xy} $. Let $ v $ be a point on line segment $ {ys_2} $ one unit from $y$ and farther from $s_2$ than from $y$, and let $ w $ be a point on line segment $ xs_1$ one unit from $x$ and farther from $s_1$ than from $x$. 
    Place {a set $ S_A $ of} $ \left\lfloor \frac{n}{5} \right\rfloor $ red points on line segment $ {xy} $ {near $ y $ (distance less than $ \frac{1}{4} $ is enough)}, and {a set $ S_B $ of} $ \left\lfloor \frac{n}{5} \right\rfloor $ blue points on line segment $ zx $ {near $ x $}. Place {a set $ S_C $ of} $ \left\lfloor \frac{n}{5} \right\rfloor $ red points on line segment $ ws_2 $ near $ w $, and {a set $ S_D $ of} $ \left\lfloor \frac{n}{5} \right\rfloor $ blue points on line segment $ vs_3 $ near $ v $. Place{, in any way, a set $ S_E$ of} the remaining $ \left\lfloor \frac{n}{5} \right\rfloor $ points on line segment $ zy $ {near $ z $}. 
    {The points must be placed according to this construction to achieve the desired bound. Figure~\ref{fig:upper} illustrates one possible configuration. Following the strategy in~\cite{urrutiaproblemas}, we} show that through any pair $ p,q $ of bichromatic points there is a disk containing {at most} $ \left\lceil \frac{n}{5} \right\rceil + 1 $ points. {If $ p,q \in S_A \cup S_B \cup S_C \cup S_E $, the result is true because of the reasoning in~\cite{urrutiaproblemas}. 
    Moreover, if $ p \in S_D $ we have three cases: (i) if $ q \in S_A $, the disk through $ p $ and $ q $ and tangent to the line $ xy $ contains at most all the points in $ S_D $, plus the point $ q $, (ii) if $ q \in S_C $, the disk through $ p $ and $ q $ and tangent to the line $ vs_3 $ contains at most all the points in $ S_C $, plus the point $ q $, and (iii) if $ q \in S_E $, the disk through $ p $ and $ q $ and tangent to the line $ zy $ contains at most all the points in $ S_D $, plus the point $ q $.}
    The points can be placed in a large enough polygon without affecting the construction.
    \begin{figure}[tb!]
	\centering
	\includegraphics{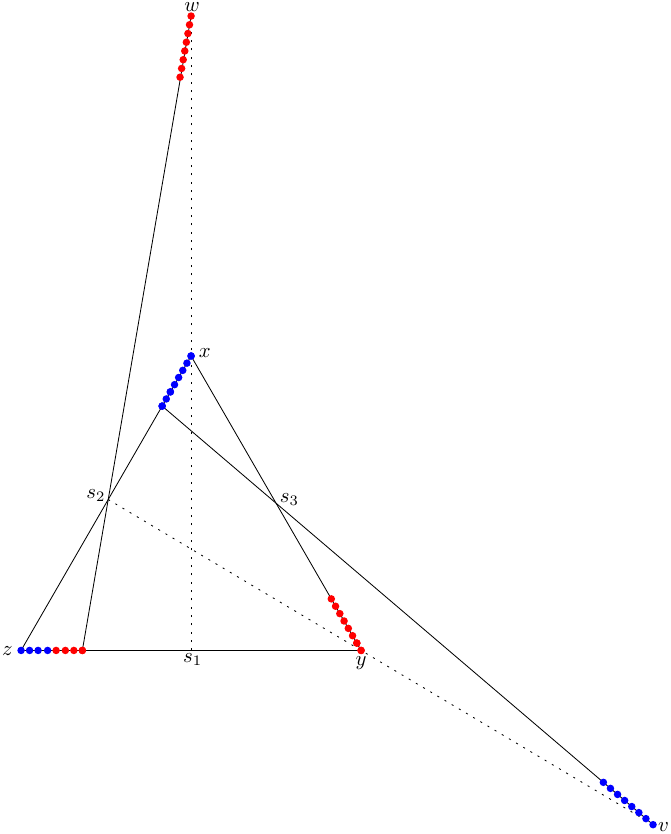}
	\caption{General bichromatic configuration. The points are enclosed in a large enough polygon (not shown).}
	\label{fig:upper}
    \end{figure}
\end{proof}

Note that this upper bound is also valid in the Euclidean setting.

\section{Conclusion} \label{sec:conclusions}

We study different variants of the following question: Given a set $ S $ of $ n $ points in a polygon~$ P $, does there always exist two points $ x, y \in S $, such that every geodesic disk containing $ x $ and $ y $  contains a constant fraction of the points of~$ S $? This question has been studied intensely in the Euclidean setting ~\cite{akiyama1996circles,claverol2021circles,edelsbrunner1989circles,hayward1989note, hayward1989some,neumann1988combinatorial,prodromou2007combinatorial, ramos2009depth,urrutiaproblemas}. We focus on the geodesic versions of this problem. An obvious open problem is the following: Can our lower bounds be improved? Basically, some of our geodesic lower bounds are not as strong as in the Euclidean case. It would be interesting to determine whether there is a separation between the bounds in the Euclidean and geodesic case. Additionally, in the Euclidean setting the best-known lower bound for the inside-outside version of the problem match the best known lower bounds for the original version of the problem exactly, whereas we proved different lower bounds in this paper. Is this difference simply an artifact of our proof or do these two versions of the problem actually have different answers? Note that even in the Euclidean case it is not known whether the two versions of the problem have different answers or not. Finally, the main open problem which has been open since 1988 is whether we can find a tight bound for the original question~\cite{neumann1988combinatorial}.

\noindent \textbf{Acknowledgments}
    G. E., D. O. and R. S. were partially supported by grant PID2023-150725NB-I00 funded by MICIU/AEI/10.13039/501100011033. P. B. and T. T. were supported in part by NSERC.

\bibliography{ref}
\end{document}